\documentclass[aps,pra,twocolumn,superscriptaddress,numerical,floatfix]{revtex4-1}

\usepackage{graphicx}
\usepackage{dcolumn}
\usepackage{bm}
\usepackage{amsmath}
\usepackage{amsthm}
\usepackage[T1]{fontenc}
\usepackage[latin9]{inputenc}
\usepackage{amssymb}
\usepackage[colorlinks]{hyperref}
\usepackage{braket}
\usepackage{placeins}
\usepackage{subcaption}
\usepackage{cleveref}

\usepackage{amsthm}
\newtheorem{theorem}{Theorem}
\newtheorem{corollary}{Corollary}

\graphicspath {{Figures/}}

\begin{document}


\title{Passive quantum error correction of linear optics networks through error averaging}


\newcommand{\affone}{Centre for Quantum Computation and Communication Technology (CQC2T), The School of Mathematics and Physics, The University of Queensland, Australia.}
\newcommand{\afftwo}{Centre for Quantum Computing \& Intelligent Systems (QCIS), University of Technology Sydney, Australia.}

\author{Ryan J. Marshman}
\affiliation{\affone}

\author{Austin P. Lund}
\affiliation{\affone}

\author{Peter P. Rohde}
\email[]{dr.rohde@gmail.com}
\homepage{http://www.peterrohde.org}
\affiliation{Centre for Quantum Software \& Information (QSI), Faculty of Engineering \& Information Technology, University of Technology Sydney, NSW 2007, Australia}
\affiliation{Hearne Institute for Theoretical Physics and Department of Physics \& Astronomy, Louisiana State University, Baton Rouge, LA 70803, United States}

\author{Timothy C. Ralph}
\affiliation{\affone}


\date{\today}

\begin{abstract}
We propose and investigate a method of error detection and noise correction for bosonic linear networks using a method of unitary averaging. The proposed error averaging does not rely on ancillary photons or control and feed-forward correction circuits, remaining entirely passive in its operation.  We construct a general mathematical framework for this technique then give a series of proof of principle examples including numerical analysis. Two methods for the construction of averaging are then compared to determine the most effective manner of implementation and probe the related error thresholds.  Finally we discuss some of the potential uses of this scheme. 
\end{abstract}


\maketitle

\section{Introduction \label{intro}}

The evolution of a multi-mode bosonic quantum state in a linear network can be simply described by a linear set of equations relating input and output bosonic modes.  These types of interactions are of interest as they are simple to arrange for most experiments involving electromagnetism but nevertheless are useful and have interesting quantum information applications.  

Linear networks are not universal for quantum information processing on their own.  However, they can be made universal using post-selection and feed forward methods with a polynomial overhead in the number of photons \cite{KLM,one-way_quantum_computer,OQC}. More recently they have been shown to deterministically generate quantum statistics that cannot be efficiently computed using classical computing resources alone (i.e.  the BosonSampling problem) \cite{Boson}. They also form the basis for optical quantum walks, for which numerous applications have been described, and been subject to widespread experimental demonstration \cite{bib:aharonov1993quantum,bib:Broome10,bib:PeruzzoQW,bib:RohdeQWExp12,bib:Schreiber10}.

Linear networks for quantum optics experiments have traditionally been implemented using bulk optical devices \cite{OQC}. 
However efforts to build integrated optical circuits have meant that the size of the networks has the potential to be made orders of magnitude smaller and consequently there is a great potential for their complexity to increase \cite{thompson2011integrated}.

In theoretical proposals for optical quantum information tasks using linear networks, it is often assumed that it is possible to configure an arbitrary linear network rapidly and with high precision.  This paper considers the second of these requirements by studying the effects of imprecision in configuring linear networks.

The model we consider assumes that large linear networks can be configured arbitrarily but with some additional noise.  This may be due to experimental imprecision of defining linear network parameters which shot-by-shot results in fluctuations of the parameters around their mean values.  We wish to concentrate on the effects due to linear network errors so we assume ideal generation of Fock basis states and the ability to make ideal Fock basis detections. Furthermore, we also assume the networks have no loss at any stage be that in the injection of states, the out-coupling to detection devices or in the network itself.  

We show that by redundantly encoding the network matrix describing a desired linear network, it is possible to generate an effect which tends towards the target network matrix when averaging over the redundant encoding.   The averaging effect occurs in a non-deterministic manner and hence the transformation acts as a filter where noise is directed into outputs which are then post-selected away (see Fig.~\ref{fig:output_probabilities}).
\begin{figure}
	\begin{centering}
		\includegraphics[width=\columnwidth]{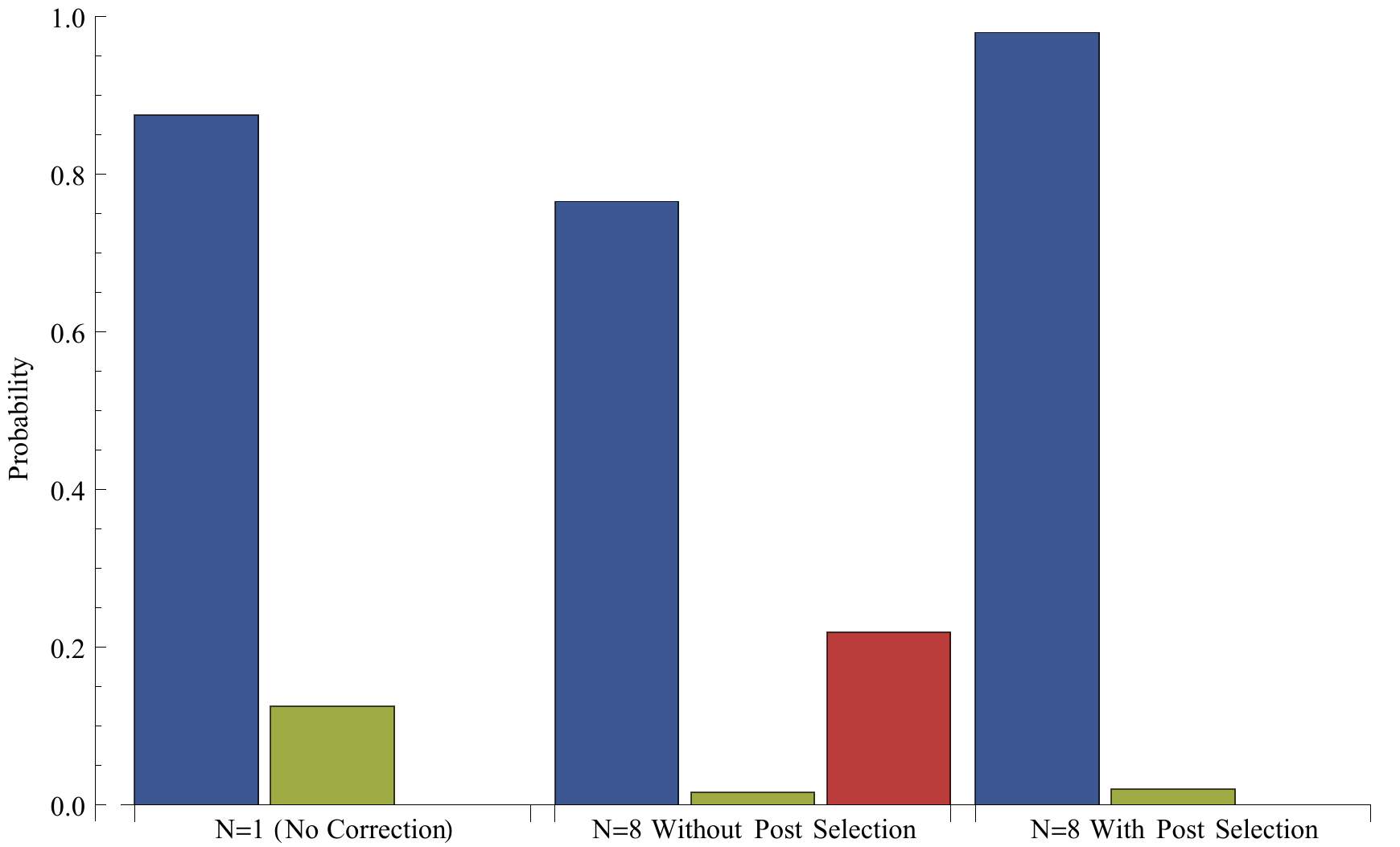}
	\end{centering}
	\caption[Comparison of output probability distributions with and without Error Averaging.]{Comparison of output probability distributions with and without error averaging. N corresponds the the number of redundant copies of the unitary being employed. Here the effect post selection has on the output distribution can be seen. The blue bars represent the probability of observing the photon in the correct output mode, green corresponds to observing the photon in the incorrect output mode and red corresponds to observing the photon in any of the error detection modes. The probabilities are based on a single photon in a Mach-Zehnder interferometer with an individual phase shifters variance $v=0.5\ \textrm{rad}^{2}$.} 
	\label{fig:output_probabilities}
\end{figure}
The central limit theorem applies to the individual matrix elements of the averaged transformation and hence their variance decreases as $\frac{1}{N}$.  The form of the average matrix and the distribution of the transformations on a finite number of averages depends on the details of the noise applied to the network encoding.  The results presented in this paper analyse these details showing conditions in which this technique may be of utility. 

The next section introduces the averaging scheme and some mathematical details that apply in the most general case. Section~\ref{implementation} includes numerous proof of principle examples which serve to highlight the effects of Error Averaging with a focus on the behaviour of the probability of success. Section~\ref{averaging at end vs step} studies two different ways of redundantly encoding a single mode phase shift and the effects of the different encodings on the resultant error and probability of success.  We then numerically analyse the averaging method for a four-mode operation in section~\ref{Four Mode Impementation Comparison}.  We will discuss some of the consequence of these results as well as future directions in Section~\ref{Discussion} and draw comparisons between Error Averaging and standard error correction in Section~\ref{Comparison with conventional error correction} before making some concluding remarks. 

\section{General Unitary Error Averaging\label{gen case}}

Here we are concerned with the case of bosonic linear scattering networks.  These are evolutions of a multi-mode bosonic field where the Heisenberg equations of motion for the annihilation operators of each mode can be written as a linear combination of all annihilation operators.  That is, if $\mathcal{U}_U$ is a unitary operation on a $m$ mode system, then
\begin{equation}
	\mathcal{U}_U a_i \mathcal{U}_U^\dagger = \sum_j U_{ij} a_j
\end{equation}
where, to preserve commutation relationships, $U$ must be a unitary matrix.  It is the network matrix $U$ that we will focus on.

Consider a linear network whose elements are those of a Discrete Fourier Transform (DFT).  That is, we have a Heisenberg style evolution between mode annihilation operators of the form
\begin{equation}
	a_{j,r} \rightarrow \frac{1}{\sqrt{N}} \sum_{k=0}^{N-1} \omega^{rk} a_{j,k}	
\end{equation}
where $\omega = e^{-i2\pi /N}$ and zero-indexing has been used, that is, $k=0$ corresponds to the first mode. The first subscript for the annihilation operator denotes the input mode and the second describes a quantity of redundancy $N$ which we explain shortly. 

We then act the $N$ copies of a target unitary $U$.  By this we mean that there is some variation between the copies but the intention was to implement the unitary $U$.  This can be described by the transformation
\begin{equation}
	a_{j,r} \rightarrow \sum_{l=0}^{m-1} (U_r)_{lj} a_{l,r}.
\end{equation}
where $N$ noisy copies of $U$ are made, denoted here as $U_1, U_2, \ldots, U_N$, where we assume an independent error model across the redundancies. 

After this the DFT matrix is applied again.  This results in the overall transformation
\begin{equation}
	a_{j,r} \rightarrow \frac{1}{N} 
	\sum_{l=0}^{m-1} \sum_{k,k^\prime=0}^{N-1}
	(U_{k^\prime})_{lj} \omega^{(r+k)k^\prime} a_{l,k}.
\end{equation}
We consider the case where all redundant modes are initialised in the vacuum state and post-select on the cases where no photons are present in the output of the redundant modes.  This means that we only need consider the parts of this transformation expression where the second subscript of the annihilation operator is zero.  In this case we have
\begin{equation}
	\label{sum_transformation}
	a_{j,0} \rightarrow \frac{1}{N}\sum_{l=0}^{m-1} \sum_{k^\prime=0}^{N-1}
	(U_{k^\prime})_{lj} a_{l,0} = \sum_{l=0}^{m-1} (M_N)_{lj} a_{l,0}
\end{equation}
where $M_N$ is a matrix defined by
\begin{equation}
	M_N = \frac{1}{N} \sum_k U_k.
\end{equation}
This matrix is then the effective linear network matrix for the post-selected system.  It includes information about the probability of success and so in general it will be not unitary.  The remainder of this paper is directed towards analysing the scenarios that arise from the multitude of choices for $U_k$ that form the expression.

For the main theorem of our work we consider a general linear network described by a unitary network matrix $U$ with any dimensionality. 


\begin{theorem}
\label{Theorem 1}
Given $N$ linear networks described by unitary matrices $\{U_1,U_2,\ldots,U_N\}$ that are random with independent and identically distributed statistics such that for all $i~\in~{1,\ldots,N}$, $\langle U_i \rangle = M$.  Then the random variable 
\begin{equation}
	\label{sum_unitary}
	M_{N}=\frac{1}{N}\sum_{i=1}^{N}U_{i}
\end{equation}
is a matrix with mean value $M$ and whose matrix elements have variance scaling as $O(1/N)$.
\end{theorem}

\begin{proof}\label{Proof 1}
Our aim in the proof is to use the central limit theorem.  Consider matrix element $r,s$ of $M_N$.  This is a random variable
\begin{equation}
	\left(M_N\right)_{rs} = \frac{1}{N} \sum_{i=1}^N \left(U_i\right)_{rs}.
\end{equation}
As the matrix elements $(U_i)_{rs}$ are constructed from unitary matrices, their magnitude is bounded by $1$.  Given this finite domain, the real and imaginary parts have maximum variance and covariance of $1$ (though these extremal values are not simultaneously achievable).  Given this bounded variance, we can use the central limit theorem to conclude that the matrix element $\left(M_N\right)_{rs}$ is a random variable with mean value $M_{rs}$.
The variance of the real or imaginary part of $(M_N)_{rs}$ is then upper bounded by $1/N$ as per the central limit theorem.
\end{proof}

The question now is what forms the mean average matrix $M$ as defined in Theorem \ref{Theorem 1}, can take. First we consider the trivial case where the unitary matrices are $1 \times 1$ dimensional.

\begin{corollary}
\label{Corollary 1}
	If each $\{U_1,\ldots,U_N\}$ are $1 \times 1$-dimensional, then $M$ is a complex number with magnitude $|M| \leq 1$.
\end{corollary}
\begin{proof}
	Write $U_k = e^{i \theta_k}$ where $p(\theta)$ is the probability density function for each of the angles $\theta_k$. From Theorem~\ref{Theorem 1} we need to compute the mean value 
	\begin{equation}
		M = \int^\pi_{-\pi} e^{i\theta} p(\theta) \mathrm{d}\theta. \label{eq:single parameter, single mode}
	\end{equation}
This is exactly the characteristic function of $p(\theta)$ evaluated at $1$.  The characteristic function is complex valued and has bounded magnitude of $1$, which is the desired result.
\end{proof}

By corollary~\ref{Corollary 1} it can be concluded that for the $1\times1$-dimensional case we can write $M=cU$ where $0 \leq c \leq 1$ and $U=e^{i\theta}$ has magnitude 1.  


Next consider higher dimensional matrices whose distribution is generated by a single parameter.  In this case, for any hermitian matrix $T$, which can be thought of as an infinitesimal generator from the $u(n)$ Lie algebra, we have
\begin{equation}
	M=\int e^{i\theta T}p(\theta)d\theta 
	\label{eq:single parameter, multi-mode}.
\end{equation}
We can make a change of variables in $\theta$ so that the distribution is changed to one that has mean zero
\begin{align}
	M&=\int e^{i(\mu + \theta^\prime)T} p(\mu + \theta^\prime) d\theta^\prime\\
	&= e^{i\mu T} \int e^{i\theta^\prime T} \bar{p}(\theta^\prime) d\theta^\prime \label{eq:seperating errors and transformations}
\end{align}
where $\bar{p}(\theta) = p(\mu + \theta)$ so that it has mean value zero.  By expanding the matrix exponential this expression can be written as
\begin{equation}
	M=\sum_n \frac{(iT)^n}{n!} \int \theta^n p(\theta) d\theta,
\end{equation}
which now relates to the moments of the underlying distribution in $\theta$. Assuming $p(\theta)$ were a Gaussian distribution with mean zero and variance $\sigma^2$ then we can write
\begin{equation}
	M = \sum_{n \in even} \frac{(iT)^n}{n!} (n-1)!! \sigma^n
\end{equation}
where $n!! = n(n-2)(n-4)\dots$ is the double factorial.  This series can be written back in the form of a matrix exponential, and by reintroducing the mean value we have
\begin{equation}
	M = e^{i\mu T} e^{-\frac{\sigma^2}{2} T^2}
\end{equation}
If $T^2=I$, which would be the case when choosing a Pauli matrix for $T$, then this expression would simplify to
\begin{equation}
M=Ue^{-\frac{\sigma^2}{2}}  \label{eq:Gaussian Psuccess}
\end{equation}
where $U$ is the unitary generated by the average parameter for $p(\theta)$.  The decaying exponential for the magnitude depends only on the variation in the distribution of $\theta$.

In the full parameter case, provided the target unitary $U$ again commutes with all errors a similar result can be found as discussed in corollary \ref{Corollary 2}.

\begin{corollary}
\label{Corollary 2}
If $\{U_1,\ldots,U_N\}$ are random $n$-dimensional unitaries such that $U_k = U exp\{i \sum_l \alpha_{kl} T_l\}$ with $n^2$ generators $T_l$ that are all hermitian and satisfy $T_l^2=I$, the parameters $\alpha_{kl}$ distributed independently with PDF $p_{l}(\alpha_l)$ which are all Gaussian with mean zero and small (but possibly different) variances so that all $U_k$ approximately commute with each other, then $M = c U$ where $0 < c < 1$ and $U$ is a unitary matrix.
\end{corollary}
\begin{proof}
We will extend the proof of Corollary~\ref{Corollary 1} to the $n$-dimensional case.  From the independence of the distributed parameters, we can write a PDF for all parameters as $p(\alpha_1,\ldots,\alpha_{n^{2}}) = p_1(\alpha_1)\times\ldots\times p_{n^2}(\alpha_{n^{2}})$.  
The approximate mutual commutivity for this expansion means 
\begin{widetext}
\newcommand{\theint}{\int^\pi_{-\pi} \ldots \int^\pi_{-\pi} \int^\pi_{-\pi}}
\newcommand{\theintd}{\mathrm{d}\alpha_1 \mathrm{d}\alpha_2 \ldots \mathrm{d}\alpha_{n^{2}}}
\begin{equation}
	\int^\pi_{-\pi} \int^\pi_{-\pi}
	[\alpha_{kl}T_{l},\alpha_{km}T_{m}] p_l(\alpha_l)p_m(\alpha_m) \mathrm{d}\alpha_l \mathrm{d}\alpha_m \approx 0 \quad \forall l,m,
\end{equation}
or in other words, that the $\alpha_{kl}T_l$ are all small with high probability.  With this we can write $M$ as
\begin{align}
	M &= U \theint
	exp\{i \sum_l \alpha_{l} T_l\} p_1(\alpha_1)\ldots p_{n^2}(\alpha_{n^{2}})
	\theintd
	\label{eq:general integral form} \\
	&\approx U \theint
	\prod_{l}exp\{i \alpha_{l} T_l\}p_1(\alpha_1)\ldots p_{n^2}(\alpha_{n^{2}})
	\theintd  \\
	&= U\int^\pi_{-\pi} exp\{i \alpha_{1} T_1\}p_1(\alpha_1) \mathrm{d}\alpha_1 	\int^\pi_{-\pi} exp\{i \alpha_{2} T_2\}p_2(\alpha_2) \mathrm{d}\alpha_2 \ldots \int^\pi_{-\pi}
	exp\{i \alpha_{n^2} T_{n^2}\}p_{n^2}(\alpha_{n^2}) \mathrm{d}\alpha_{n^2} \\
	&\approx U \prod_l e^{-\frac{\sigma_l^2 T_l^2}{2}},\label{eq:approx commuting, general case}
\end{align}
\end{widetext}
where $\sigma_l^2$ is the variance of $p_l$ and the final approximation is assuming the distribution is small so that the bounds of the integration do not matter.  Using the $T_l^2=I$ requirement on the generators the final product of exponentials can be identified with the value $c$, we have the desired result.
\end{proof}
The requirement of $T_L^2=I$ merely reflects a simplification where the generators are built from the Pauli matricies which are the constructions we will focus on in this paper.  If this is not the case then it is possible to identify the hermitian operator $\prod_l e^{-\frac{\sigma_l^2 T_l^2}{2}}$ as a state dependant decay in the amplitude of the operator.



Finding expressions for the matrix $M$ outside of the situations just outlined is an open problem.  In the most general case, $M$ is not proportional to a unitary matrix.  Furthermore, is not guaranteed that $M$ will satisfy the conditions for a normal matrix and hence cannot be unitary diagonalised.  So it is unclear if in general this post-selected regime has any connection to unitary quantum evolution at all.  Nevertheless, we will begin to examine situations which approach this domain through decompositions into single parameter problems and using numerical computations.

\section{Implementation\label{implementation}}

This section demonstrates how Error Averaging can be implemented for various example optical systems. These examples also serve as a verification of the range of validity of the approximately commuting errors assumption. It can also be noted that Equation~\ref{sum_unitary} can become the appropriate transformation for duality quantum computing by allowing the $U_{i}$ to be arbitrary\cite{dualityQC}.

Constructions for the redundant encoding using the DFT implementation from the previous section are useful mathematically but may be inconvenient to implement in practice.  The transformation of Eq.~\ref{sum_transformation} can also be achieved using an array of beam-splitters as shown in Fig.~\ref{fig:gen system}.  This beam-splitter array has the desirable property of being generated by a recursive pattern.  As shown by the bounding rectangles in Fig.~\ref{fig:gen system}, the outer and inner layers share the same basic structure. 

\begin{figure}[tbh]
	\includegraphics[width=\columnwidth]{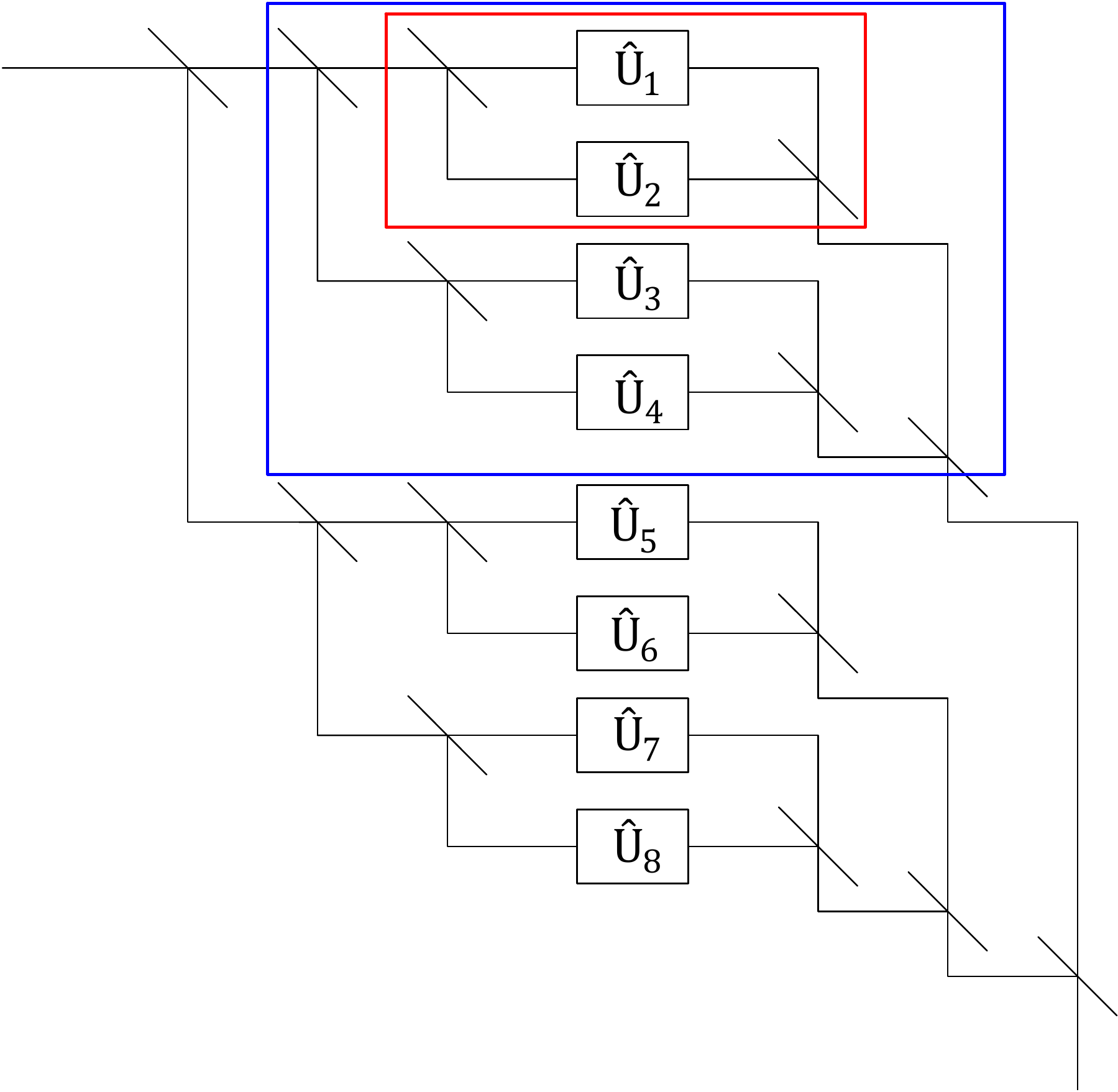}
	\caption{\label{fig:gen system}Redundant encoding using 50:50 beam-splitters for $N=8$. The boxes labelled $\hat{U}_i$ can be single or multi-mode.  In the multi-mode case the encoding beam-splitter network is repeated for each mode input and output. Output modes that are post-selected on the vacuum are not shown here. The red box shows an $N=2$ level encoding and the blue box shows $N=4$.  Further nesting of this arrangement can achieve any $N$ being a power of two.}
\end{figure}
	
	


All linear networks can be generated by arranging networks of beam-splitters and phase shifts~\cite{reck}.  Carolan et.~al.~\cite{ULO} have experimentally probed a linear network where all possible networks can be generated using controllable phase shifts and unvarying beam-splitters.  In their experimental implementation they demonstrated the ability to implement many quantum logic gates and linear optical protocols with a high fidelity.  Following this same methodology one can generate controllable beam-splitters using a Mach-Zehnder (MZ) interferometer consisting of a controllable phase shift in one arm and two fixed $50:50$ beam splitters.

Within this type of architecture the controllable phase shift is the key source of non-systematic noise.  Furthermore, redundantly encoding phase shifts are well characterised by the results presented above from Corollary~\ref{Corollary 1}.  So we will focus on phase shift induced errors for the analysis of this section and the next.  The model we will use assumes $50:50$ beam-splitters which are fixed and phase shifts that vary and are the source of all noise.

The noise in a controllable phase-shift can be written as $e^{i(\theta+\delta)}$ where $\theta$ is a real number representing the phase shift to be applied and $\delta$ is a zero-mean random variable representing the error.  For the identity operation $\theta=0$.  We will assume the distribution for $\delta$ to be Gaussian with variance $v$.  For values of $v$ that are comparable to $\pi^2$ the multi-valued nature of phase shifts becomes important.  But initially we will focus on the limit where $v \ll \pi^2$.



The remainder of this section considers the above implementation of a tunable beam-splitter as a MZ interferometer with the phase shift being error averaged. The error averaging will be performed using the concatenated beam-splitter network, hence $N=2^n, n \in \mathbb{N}$ and all beam-splitters used in this system will be fixed and with a splitting ratio of 50:50. We will analyse two key cases, the single photon and two photon performance.  The former involves the classical wave nature of the probability distribution for a single photon.  The latter includes Hong-Ou-Mandel~\cite{hom} style quantum interference.

\subsection{1 photon inputs \label{1 photon N arbitrary}}

The 1 photon network considered here is shown in Figure~\ref{fig:MZ_setup} both without any correction \ref{fig:uncorrected MZ beam splitter} and for the $N=2$ case \ref{fig:corrected MZ beam splitter}.  
\begin{figure}[tbh]
	\begin{subfigure}[b]{0.8\columnwidth}
		\includegraphics[width=1\columnwidth]{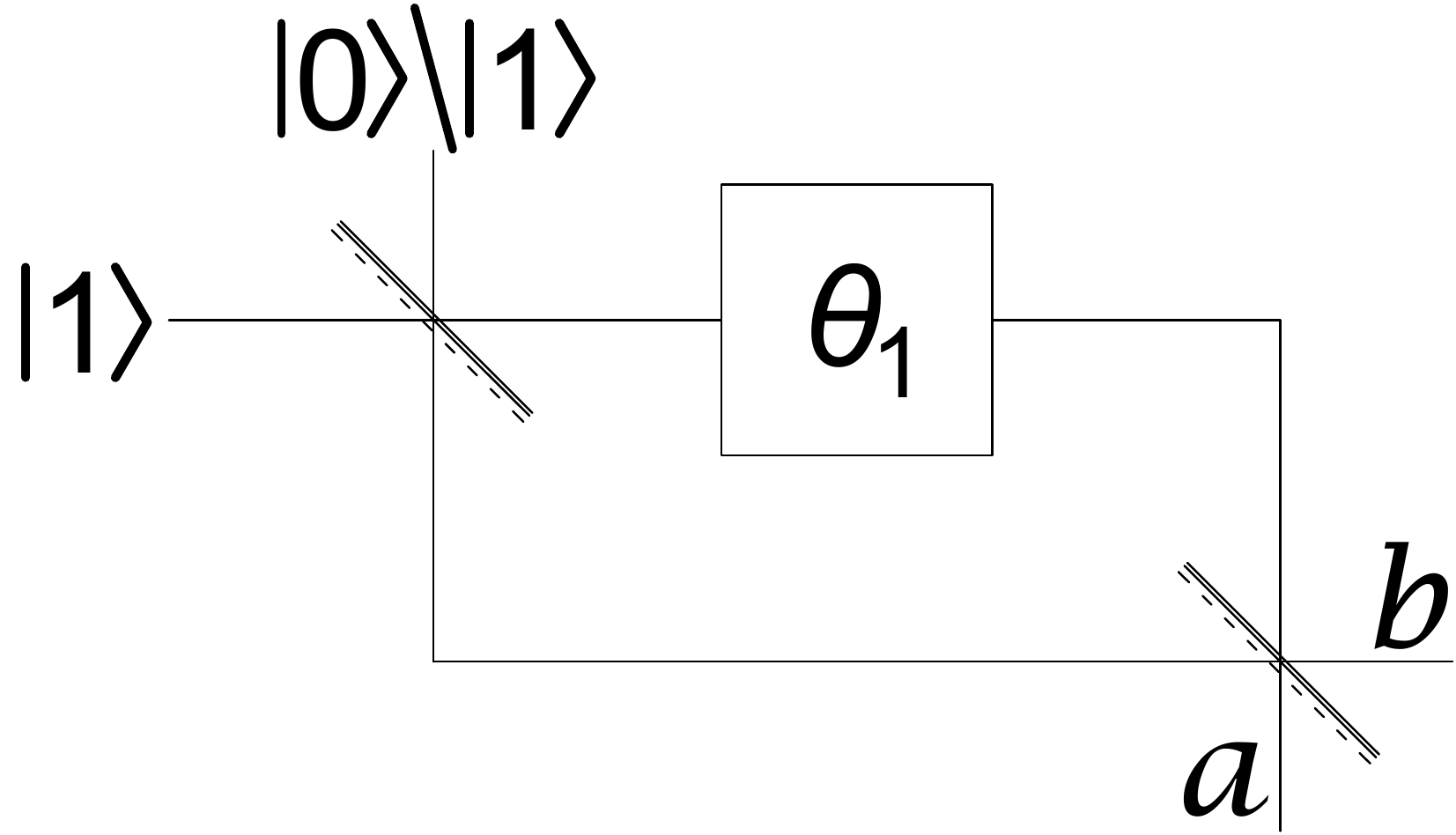}
		\caption{}
		\label{fig:uncorrected MZ beam splitter} 
	\end{subfigure}
	
	\begin{subfigure}[b]{0.9\columnwidth}
		\includegraphics[width=\columnwidth]{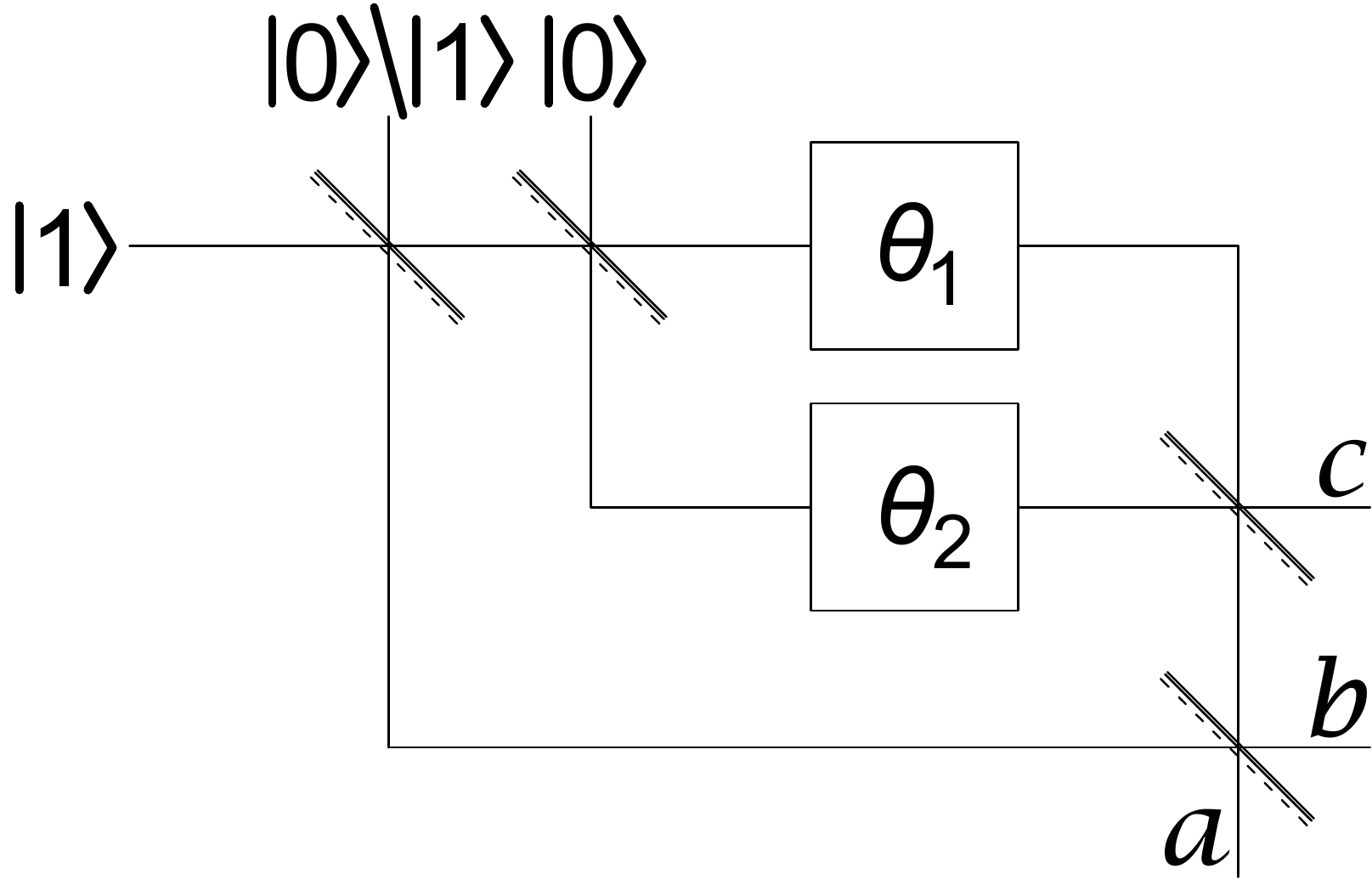}
		\caption{}
		\label{fig:corrected MZ beam splitter}
	\end{subfigure}
			\caption{\label{fig:MZ_setup}Diagram of MZ based tunable beam-splitter. (a) An uncorrected beam splitter implemented via MZ interferometer and (b) such a beam splitter corrected by redundantly encoding the phase shift, here for $N=2$. $a$ and $b$ label output modes and $c$ labels an error detection mode. The input state shown is used for both one and two photon calculations. The phase shift elements are marked with $\theta_{j}$ and are random variables.}
\end{figure}
	
The input single photon state is $\left|\phi\right\rangle = \hat{a}^{\dagger}\left|0\right\rangle $.  After traversing the error averaged network, the resulting un-normalised output state conditional on all encoded modes being vacuum is
\begin{multline}
	\ket{\psi} = \left(
	\left( \frac{e^{i\theta}}{2} \left\{\frac{1}{N}\sum_{j=1}^{N}e^{i\delta_{j}}\right\} + \frac{1}{2}\right) \hat{a}^{\dagger}  \right. \\
	+ \left. \left(\frac{e^{i\theta}}{2} \left\{\frac{1}{N}\sum_{j=1}^{N}e^{i\delta_{j}}\right\} - \frac{1}{2}\right) \hat{b}^{\dagger} 
	\right)
	\ket{0}.
	\label{eq:1ParbN}
\end{multline}
which is consistent with Theorem~\ref{Theorem 1}. Here $\theta_{j}=\theta+\delta_{j}$ with $\theta$ a constant and $\delta_j$ a random variable. 

As linear networks conserve photon number, and we have post-selected the cases where energy exits via the redundant encoding modes, we know that the output state always contains one and only one photon.  The probability that the photon is measured in a particular mode can therefore be equated to the average photon number in that mode.  Using this we can calculate from the un-normalised state $\ket{\psi}$ the probability of observing the photon in the $\hat{a}$ and $\hat{b}$ modes without post-selection to be 
\begin{eqnarray}
	\bra{\psi} \hat{a}^{\dagger}\hat{a} \ket{\psi} & \approx & \cos^{2}\left(\theta/2\right)+\frac{v}{4N}-\frac{v\cos^{2}(\theta/2)}{2}
\end{eqnarray}
and
\begin{eqnarray}
	\left\langle \psi\right|\hat{b}^{\dagger}\hat{b}\left|\psi\right\rangle & \approx & \sin^{2}\left(\theta/2\right)+\frac{v}{4N}-\frac{v\sin^{2}(\theta/2)}{2}
\end{eqnarray}
where we have taken a first-order expansion in the phase shift variance $v$.
The probability of success is the sum of the probabilities of the $\hat{a}$ mode and $\hat{b}$ mode.  This is 
\begin{eqnarray}
	P(\textrm{success}) & = & \bra{\psi}\hat{a}^{\dagger}\hat{a}\ket{\psi} +\bra{\psi}\hat{b}^{\dagger}\hat{b}\ket{\psi} \\
	& \approx & 1+\frac{v}{2N}-\frac{v}{2}. \label{eq:1pNarbitrary successs}
\end{eqnarray}
In the large $N$ limit, this corresponds to the linear approximation of Equation \ref{eq:Gaussian Psuccess} where we can identify $P(\textrm{success})=c$. 

Without any noise, choosing $\theta=0$ results in complete interference and the input single photon state will be transferred to a single output.  Any deviations from this are attributed to non-ideal interferometer performance. 
In this case the probability of observing the output in the correct mode without post-selection is
\begin{equation}
	\left\langle \psi\right|\hat{a}^{\dagger}\hat{a}\left|\psi\right\rangle \approx 1 - \frac{(2N-1)v}{4N}. \label{eq:1pNoPost}
\end{equation}
After post-selection this becomes
\begin{equation}
	\frac{\left\langle \psi\right|\hat{a}^{\dagger}\hat{a}\left|\psi\right\rangle}{P(\textrm{success})} \approx 1 - \frac{v}{4N}. \label{eq:1pWithPost}
\end{equation}
		
Figure~\ref{fig:post vs no post} shows how these two quantities scale with $N$. In particular, it can be seen that after post-selection the likelihood of the photon exiting the interferometer in the correct mode can be made arbitrarily close to unity by increasing $N$. Also, while the probability of success decreases for increasing $N$ it asymptotes to a constant value. This implies that as $N$ increases, even though the total quantity of errors added to the system increases, the effects of the combined errors on the interferometer is less. This result is also not dependent on the value chosen for $\theta$. Explicitly as the intended phase shift can be factored out in Eq. \ref{eq:1ParbN} similarly to the result shown in Eq. \ref{eq:Gaussian Psuccess} the effects of errors and our error correction can be considered separately from the transformation being applied.
\begin{figure}[tbh]
	\includegraphics[width=\columnwidth]{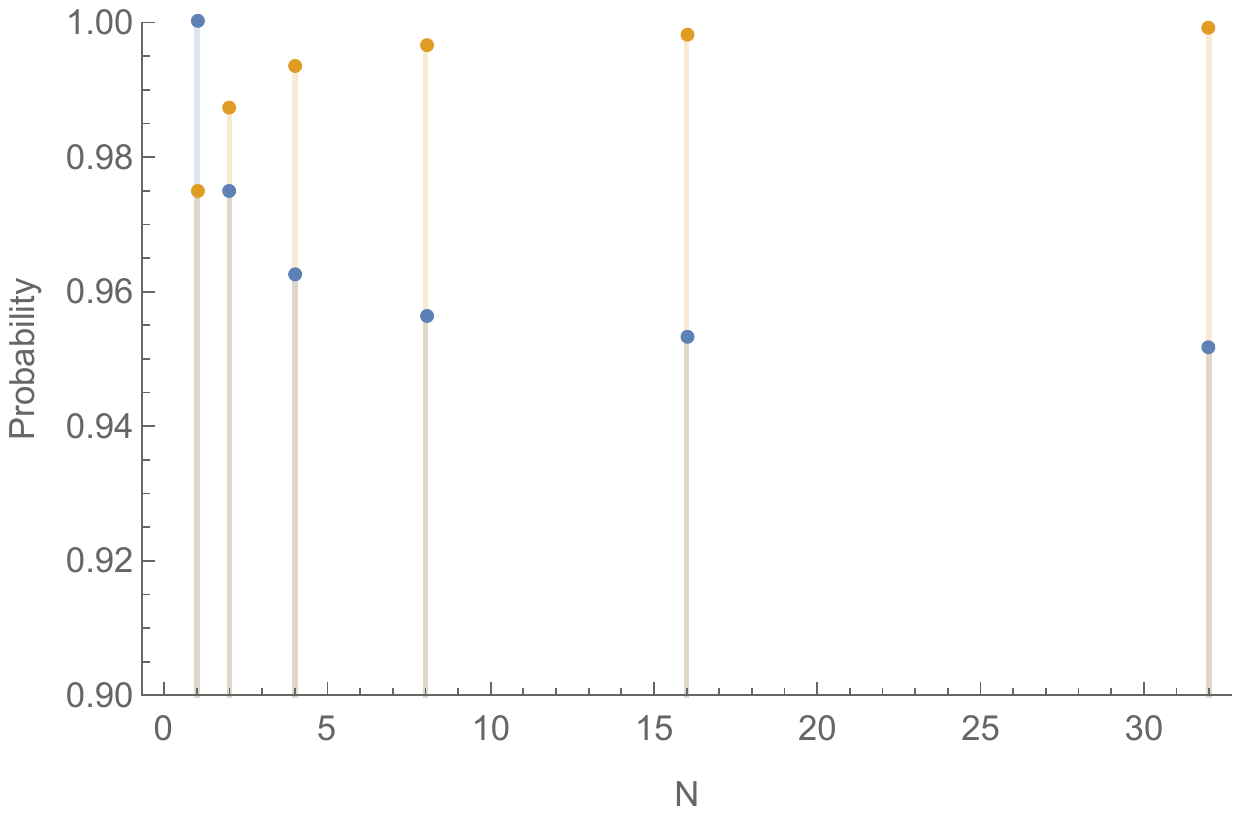}
	\caption{\label{fig:post vs no post} Probability of single photon being detected at the $a$ output port as shown in Figure~\ref{fig:MZ_setup} as a function of the redundant encoding size $N=2^n$. Here the phase shifts are sampled from a distribution with mean value $0$ and variance $v=0.1\ \textrm{rad}^{2}$. The blue values give the probability of success and the orange values corresponds to the probability of obtaining the correct result conditional on the photon not exiting the added redundant encoding modes, that is, with post-selection.  Eq.~(\ref{eq:1pNarbitrary successs}) predicts an asymptote of $0.95$ without post-selection and Eq.~(\ref{eq:1pWithPost}) predicts an asymptote of $1$ with post-selection.  The asymptotic behaviour is consistent with the plotted data.}
\end{figure}

\subsection{2 photon inputs\label{2 photons N arbitrary}}

The single photon interference effects in linear networks can be explained using classical wave interference.  Now we will consider two photon interference to demonstrate the behaviour of quantum interference when using the redundant encoding. As such here $\left|1,1\right\rangle$ is used as the input state. Again, a diagram of the explicit set-up with and without the redundant encoding can be seen in Figure \ref{fig:MZ_setup}. For two photons the un-normalised output state for the $a$ and $b$ modes is
\begin{widetext}
\begin{equation}
	\ket{\psi} =  \frac{1}{2}\left\{ 1+\frac{1}{N^{2}}\left(\sum_{j=1}^{N}\sum_{k=1}^{N}e^{i(\delta_{j}+\delta_{k})}\right)\right\} \ket{1,1} \\
	+\frac{\sqrt{2}}{4}\left\{ \frac{1}{N^{2}}\left(\sum_{j=1}^{N}\sum_{k=1}^{N}e^{i(\delta_{j}+\delta_{k})}\right)-1\right\} \left(\ket{2,0} +\ket{0,2}\right)\label{eq:2pNarbitrary S}
\end{equation}
\end{widetext}
where we have chosen $\theta=0$ when computing this state. This is done, as above, to simplify the form of the equations and does not change the effect of the redundant encoding on the errors.  Because of this choice, the action of the interferometer on the input state should be the identity operation and hence $\ket{1,1}$ is the desired output state.   Note that we could have chosen the input state to be $\ket{2,0}$, but this would not necessarily show any new behaviour, just the single photon results independently applied to the two input photons. 

We can again write probabilities as expectation values of occupation number.  Using the form of Eq.~(\ref{eq:2pNarbitrary S}), the ideal output is achieved when
\begin{equation}
\bra{\psi}\hat{a}^{\dagger}\hat{a}\hat{b}^{\dagger}\hat{b}\ket{\psi} = 1.
\end{equation}
This expectation value for the state including the phase shift noise is
\begin{eqnarray}
\left\langle \psi\right|\hat{a}^{\dagger}\hat{a}\hat{b}^{\dagger}\hat{b}\left|\psi\right\rangle & = & \left\langle \left|\frac{1}{2}\left\{ 1+\frac{1}{N^{2}}\left(\sum_{j=1}^{N}\sum_{k=1}^{N}e^{i(\delta_{j}+\delta_{k})}\right)\right\} \right|^{2}\right\rangle \nonumber \\
& = & \left\langle \frac{1}{4}\left(1+\frac{2}{N^{2}}\left(\sum_{j=1}^{N}\sum_{k=1}^{N}\cos\left(\delta_{j}+\delta_{k}\right)\right)\right)\right\rangle \nonumber \\
&  & +\frac{1}{4}\Biggl\langle\frac{1}{N^{4}}\left(\sum_{j}^{N}e^{-2i\delta_{j}}+\sum_{j=1}^{N}\sum_{k\ne j}^{N}e^{-i(\delta_{j}+\delta_{k})}\right)\nonumber \\
&  & \times\left(\sum_{l}^{N}e^{2i\delta_{l}}+\sum_{l=1}^{N}\sum_{m\ne l}^{N}e^{i(\delta_{l}+\delta_{m})}\right)\Biggr\rangle\nonumber \\
& \approx & 1-v\label{eq:exp. value aabb}
\end{eqnarray}
where the approximation is assuming $v$ small.  Post-selection will increase this to 
\begin{widetext}
\begin{eqnarray}
P(\textrm{coincidence}) & = & \frac{\left\langle \psi\right|\hat{a}^{\dagger}\hat{a}\hat{b}^{\dagger}\hat{b}\left|\psi\right\rangle }{\left\langle \psi\right|\hat{a}^{\dagger}\hat{a}\hat{b}^{\dagger}\hat{b}\left|\psi\right\rangle +0.5\left\langle \psi\right|\hat{a}^{\dagger}\hat{a}^{\dagger}\hat{a}\hat{a}\left|\psi\right\rangle +0.5\left\langle \psi\right|\hat{b}^{\dagger}\hat{b}^{\dagger}\hat{b}\hat{b}\left|\psi\right\rangle }\nonumber \\
& \approx & 1-\frac{v}{2N}\label{eq:2pNarb PS}
\end{eqnarray}
\end{widetext}
where $P(\textrm{coincidence})$ is the probability the photons exit modes $a$ and $b$ individually and the binomial approximation has been used to keep only variance terms to first order.  Finally the probability of success, which is the probability no photons exit the encoding modes is 
\begin{eqnarray}
	P(\textrm{success}) & \approx & 1-v+\frac{v}{2N}.\label{eq:2pNarb Success}
\end{eqnarray}
Again, the large $N$ limit of this equation the result matches the prediction of Equation~\ref{eq:Gaussian Psuccess}. Also, the probability of success has a $\frac{1}{N}$ scaling which is the same as for the single photon input case. 

In this section we have demonstrated how redundantly encoding variable components can reduce the resulting variance within a system for the simple but highly important case of a single beam splitter. We have also shown that the results match what is expected from the couple of solved exact cases discussed in Section \ref{gen case}. In the following section we will give some more complex examples to give a clearer insight into how this redundant encoding might best be applied, and its effect in the situations where the mathematical machinery introduced earlier is not easily solvable.

\section{Comparison between averaging techniques \label{averaging at end vs step}}

In this section we will study two different methods, which we will refer to as \textit{averaging at the end} and \textit{averaging each step}. To illustrate the two approaches we consider a simple system of phase-shifters. Fig.~\ref{fig:Different methods of implementation} shows schematically these two configurations as well as a baseline comparison. The system analysed is applying a single mode phase-shift generated by $M$ sequential phase shifters.  Averaging across the entire system applies the $M$ phases and redundantly encoding this $N$ times (Figure~\ref{fig:Different methods of implementation b}). The method of averaging each step involves a redundancy of $N$ for each of the $M$ applied phase shifts (Figure~\ref{fig:Different methods of implementation c}).  When averaging each component individually significantly more encoding beam-splitters are required, however we will show this leads to more stability in the output state for larger errors. In the low error limit however these two methods yield equivalent results. Because of this the difference in clearer when results are taken to the the higher order and as such in the following section all approximations will be taken to the second order in the variance as opposed to the first order as done above.

\begin{figure}
	\centering
	\begin{subfigure}[b]{\columnwidth}
		\includegraphics[width=0.65\columnwidth]{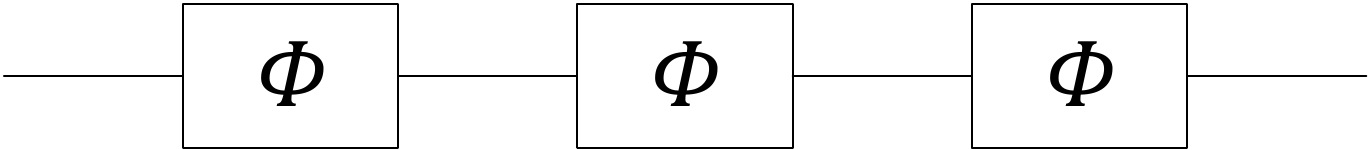}
		\caption{}
		\label{fig:Different methods of implementation a} 
	\end{subfigure}
	
	\begin{subfigure}[b]{\columnwidth}
		\includegraphics[width=0.80\columnwidth]{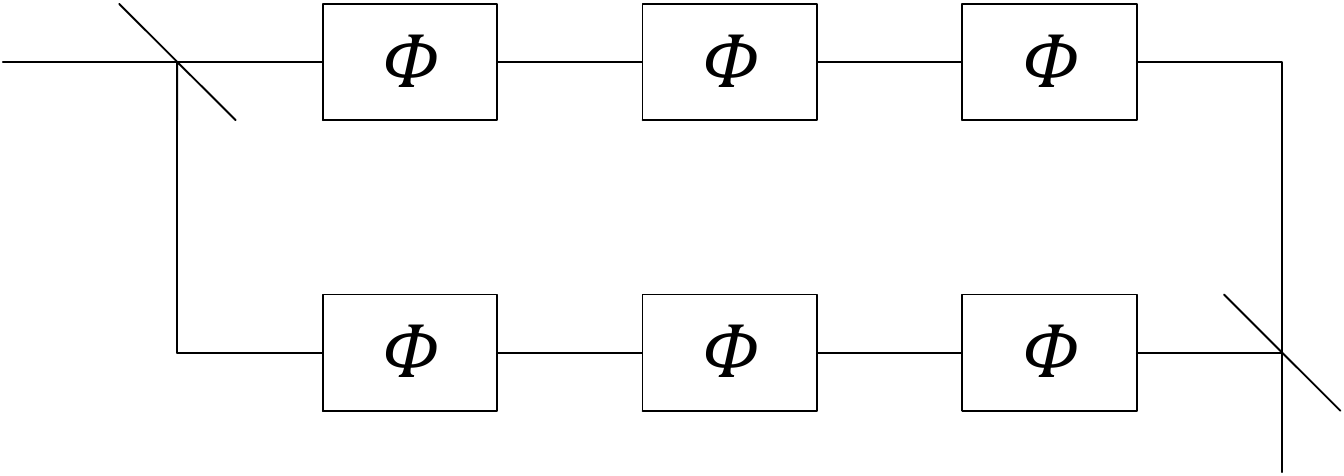}
		\caption{}
		\label{fig:Different methods of implementation b}
	\end{subfigure}

	\begin{subfigure}[b]{\columnwidth}
		\includegraphics[width=1\columnwidth]{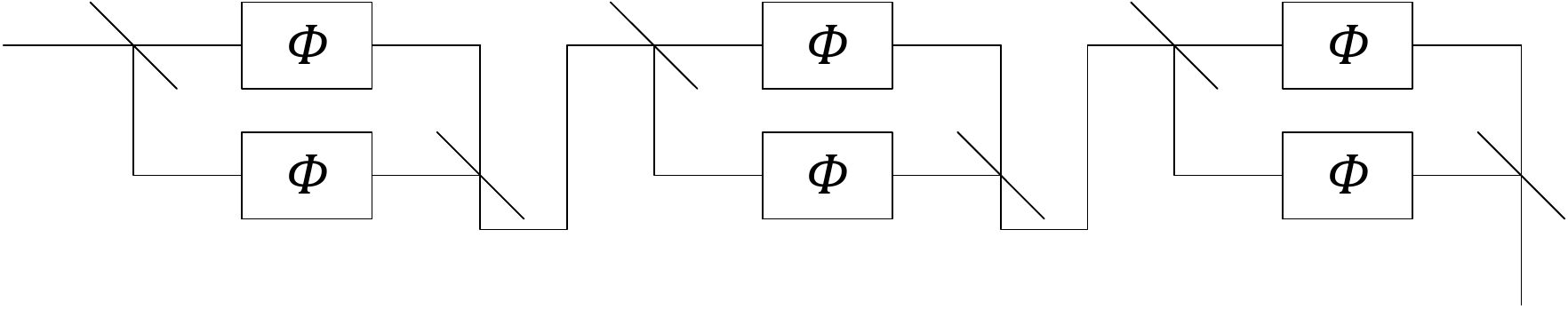}
		\caption{}
		\label{fig:Different methods of implementation c}
	\end{subfigure}
	
	\caption[Two numerical solutions]{Three methods of applying three phase shifts, each marked with a in
			series. (a) Three phase shifts with no error averaging. (b) Three phase
			shifts when averaging across the system. (c) Three phase shifts when
			averaging across each phase shifter individually. Averaging across the
			system will in general require far fewer encoding resources.
			\label{fig:Different methods of implementation}}
\end{figure}

Following an approach motivated by the previous section the applied phase shifters were placed in one arm of a MZ interferometer. The applied phase shift was chosen to have mean zero with a Gaussian random noise with a variance $v$. This choice allows for the errors here to be compared with those modelled in sections~\ref{1 photon N arbitrary} and~\ref{2 photons N arbitrary}. The probability of success is now defined as the photon number expectation value evaluated at the end of the phase applying systems while the strength of the error is defined as the photon number expectation value of the MZ interferometer system at the expected output given no photon exited any of the redundant encoding modes. 


\subsection{No Averaging\label{No Averaging}}

Starting with the baseline comparison case where no Error Averaging is used (Fig.~\ref{fig:Different methods of implementation a}), the output state for a single photon going through $M$ phase shifters will be
\begin{equation}
\left|\psi\right\rangle =\left(\prod_{k=1}^{M}e^{i\delta_{k}}\right)\left|1\right\rangle \label{eq:noAvPhaseState}
\end{equation}
As there is no path for the photon to exit the system, the probability of success is always $1$.

To quantify the error the phase applying system it was then inserted into a MZ interferometer giving a total output state $\ket{\Psi}$.  As the mean phase shift is zero, the error is manifest in the photon expectation value at the correct output mode after post-selection. As however the probability of success is $1$ no post selection occurs here.  The output state from the MZ interferometer is
\begin{equation}
	\ket{\Psi} =\frac{1}{2}\left(\prod_{k=1}^{M}e^{i\delta_{k}}+1\right)\hat{a}^{\dagger}\ket{0}+\left(\prod_{k=1}^{M}e^{i\delta_{k}}-1\right)\hat{b}^{\dagger}\ket{0}. \label{eq:noAveIntState}
\end{equation}
The measure to the quantity of error which was used to compare the three situations chosen is the probability of observing the correct result conditional on the photon not being detected in an error mode, or $P(\textrm{correct})$. For no averaging this is
\begin{eqnarray}
P(\textrm{correct}) & = & \left\langle \Psi\right|\hat{n}_{a}\left|\Psi\right\rangle \nonumber \\
& = & \frac{1}{2}\left\langle 1+\cos\left(\alpha\right)\right\rangle \nonumber \\
& \approx & 1-\frac{Mv}{4}+\frac{M^{2}v^{2}}{16}\label{eq:ErrorNoAv1}
\end{eqnarray}
where $\alpha=\sum_{k=1}^{M}\delta_{k}$ and Gaussian statistics have been used to write higher order moments in terms of the variance.

\subsection{Averaging Across the Entire Phase System\label{Averaging Across the Entire Phase System}}

We now consider averaging across the whole system, as shown in Figure~\ref{fig:Different methods of implementation b}. Proceeding as before, the state for a single photon after passing through $M$ phase shifters in series which is being averaged across $N$ times will simply be
\begin{equation}
	\ket{\psi}=\frac{1}{N}\sum_{j=1}^{N}\left(\prod_{k=1}^{M}e^{i\delta_{j,k}}\right)\ket{1} \label{eq:AvEndPhaseState}
\end{equation}
The probability of success is thus
\begin{eqnarray}
	P\left(success\right) & = & \braket{\psi|\psi} \nonumber \\
& \approx & \left[1-\left(1-\frac{1}{N}\right)\left(Mv-\frac{1}{2}M^{2}v^{2}\right)\right]\label{eq:AveEndProbSuccess}
\end{eqnarray}
This result is similar to the what was found in previous sections, see Eq.\ref{eq:1pNarbitrary successs} and Eq.\ref{eq:2pNarb Success}, with the probability of success asymptotically approaching some fixed value for large $N$.

To determine the size of the error, the phase applying system was again inserted into one arm of a MZ interferometer giving a total output state $\ket{\Psi}$. The error is then given by the photon expectation value in the correct output mode with post selection. The output state is
\begin{eqnarray}
	\ket{\Psi}& = &\frac{1}{2}\left(\frac{1}{N}\sum_{j=1}^{N}\left(\prod_{k=1}^{M}e^{i\delta_{j,k}}\right)+1\right)\hat{a}^{\dagger}\ket{0}\\ & & +\left(\frac{1}{N}\sum_{j=1}^{N}\left(\prod_{k=1}^{M}e^{i\delta_{j,k}}\right)-1\right)\hat{b}^{\dagger}\ket{0}\label{eq:AveEndIntState}
\end{eqnarray}
So the photon number expectation value for the expected output from the interferometer will be
\begin{widetext}
\begin{eqnarray}
	\bra{\Psi}\hat{n}_{a}\ket{\Psi}
& \approx & 1-\frac{1}{4}\left(Mv-\frac{M^{2}v^{2}}{4}+\left(1-\frac{1}{N}\right)\left(Mv-\frac{1}{2}M^{2}v^{2}\right)\right)
\end{eqnarray}
Similarly for the incorrect output port, the photon number expectation value will be
\begin{equation}
	\bra{\Psi}\hat{n}_{b}\ket{\Psi} = \frac{1}{4}\left\langle 1+\left\langle \psi|\psi\right\rangle -\frac{2}{N}\sum_{j=1}^{N}\cos\left(\alpha_{j}\right)\right\rangle 
\end{equation}
Therefore, our error measure, the conditional probability of observing the correct result, will now be
\begin{equation}
P(\textrm{correct})  \approx  \left[1-\frac{1}{4}\left(Mv-\frac{M^{2}v^{2}}{4}+\left(1-\frac{1}{N}\right)\left(Mv-\frac{1}{2}M^{2}v^{2}\right)\right)\right]\nonumber\times\left[1-\left(1-\frac{1}{N}\right)\left(\frac{Mv}{2}-\frac{1}{4}M^{2}v^{2}\right)\right]^{-1}\label{eq:ErrorAvEnd}
\end{equation}
\end{widetext}

\subsection{Averaging Across Each Phase Shifter Individually\label{Averaging Across Each Phase Shifter Individually}}

If each phase shifter is averaged individually, as seen in Figure \ref{fig:Different methods of implementation c}, then the state for a single photon after passing through the phase applying system will be
\begin{equation}
	\ket{\psi}=\frac{1}{N}\prod_{k=1}^{M}\left(\sum_{j=1}^{N}e^{i\delta_{j,k}}\right)\ket{1}\label{eq:AveStepPhaseState}
\end{equation}
Reproducing the above calculations with this state yields a probability of success of
\begin{equation}
P\left(Success\right)\approx\left(1-\left(v-\frac{v^{2}}{2}\right)\left(1-\frac{1}{N}\right)\right)^{M}\label{eq:AvStepProbSuccess}
\end{equation}
and a conditional probability of observing the correct result of
\begin{widetext}
\begin{equation}
P(\textrm{correct}) \approx  \left[\frac{3}{4}-\frac{Mv}{4}+\frac{M^{2}v^{2}}{16}+\frac{1}{4}\left(1-\left(v-\frac{v^{2}}{2}\right)\left(1-\frac{1}{N}\right)\right)^{M}\right]\nonumber\times\left[\frac{1}{2}+\frac{1}{2}\left(1-\left(v-\frac{v^{2}}{2}\right)\left(1-\frac{1}{N}\right)\right)^{M}\right]^{-1}\label{eq:ErrorAvStep}
\end{equation}
\end{widetext}
Importantly, for both this case as well as when averaging each step, if only the first order approximation is used and $M=1$ then the error matches the error found in section \ref{1 photon N arbitrary}. However we see that with the second order terms included the two results diverge from one another. This can be seen most clearly in Figure \ref{fig:Probability-of-success all}.

\subsection{Summary of Errors and Probabilities\label{Summary of Errors and Probabilities}}

\begin{figure}
	\centering
	\begin{subfigure}[b]{\columnwidth}
		\includegraphics[width=1\columnwidth]{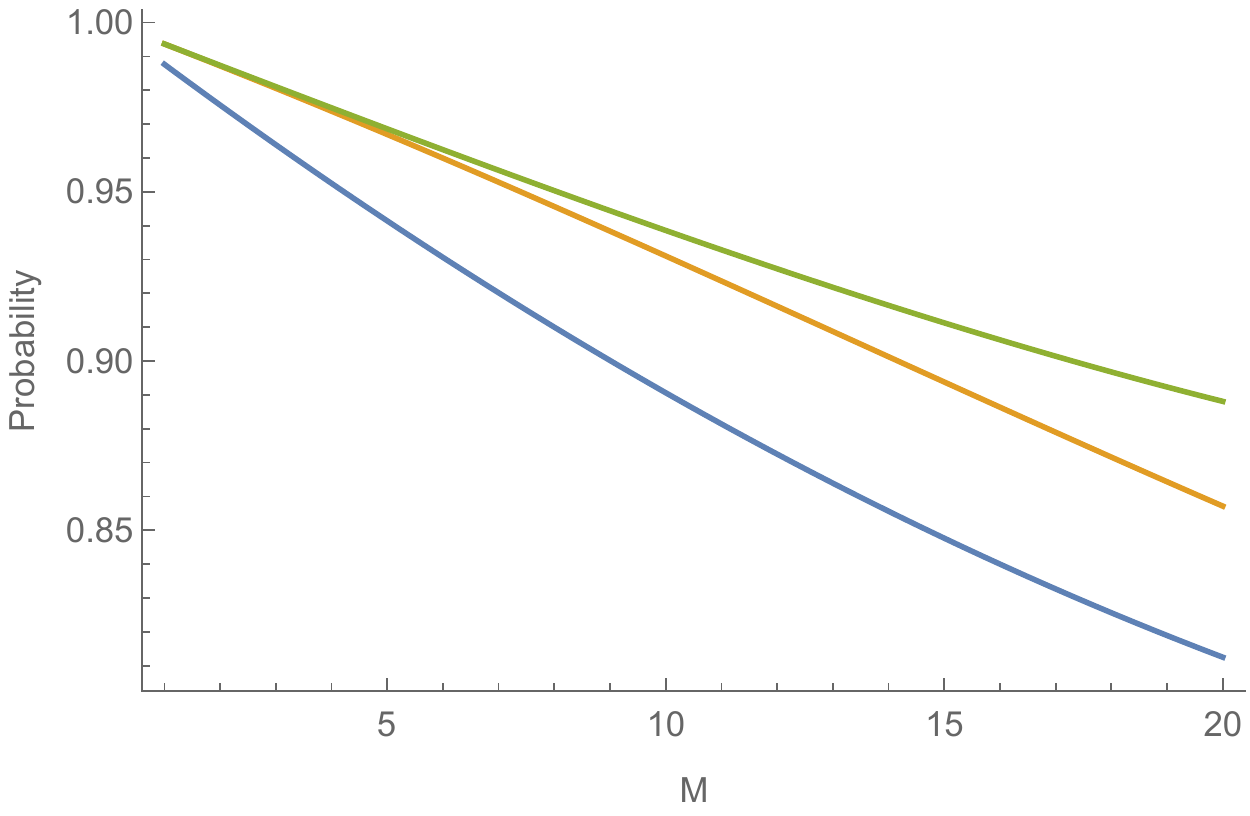}
		\caption{}
		\label{fig:Prob correct for phase systems a} 
	\end{subfigure}
	
	\begin{subfigure}[b]{\columnwidth}
		\includegraphics[width=1\columnwidth]{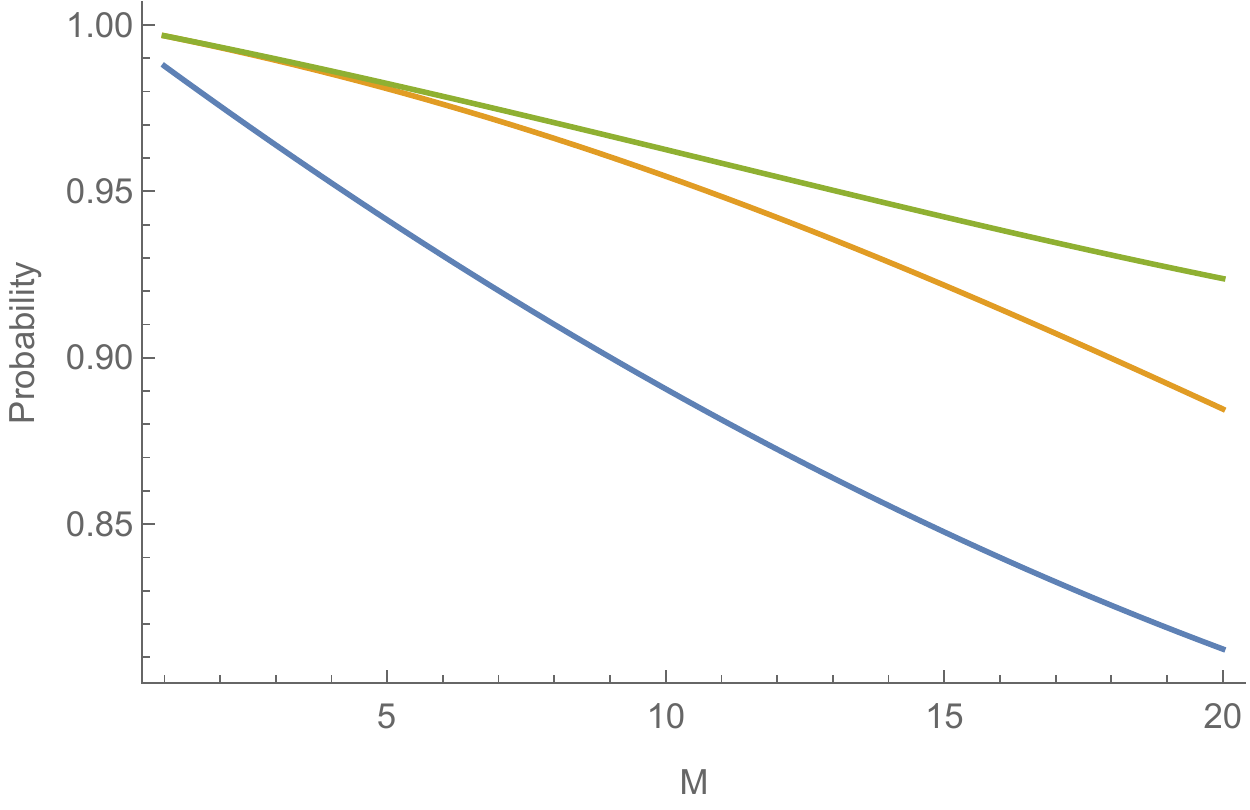}
		\caption{}
		\label{fig:Prob correct for phase systems b}
	\end{subfigure}
	
	\begin{subfigure}[b]{\columnwidth}
		\includegraphics[width=1\columnwidth]{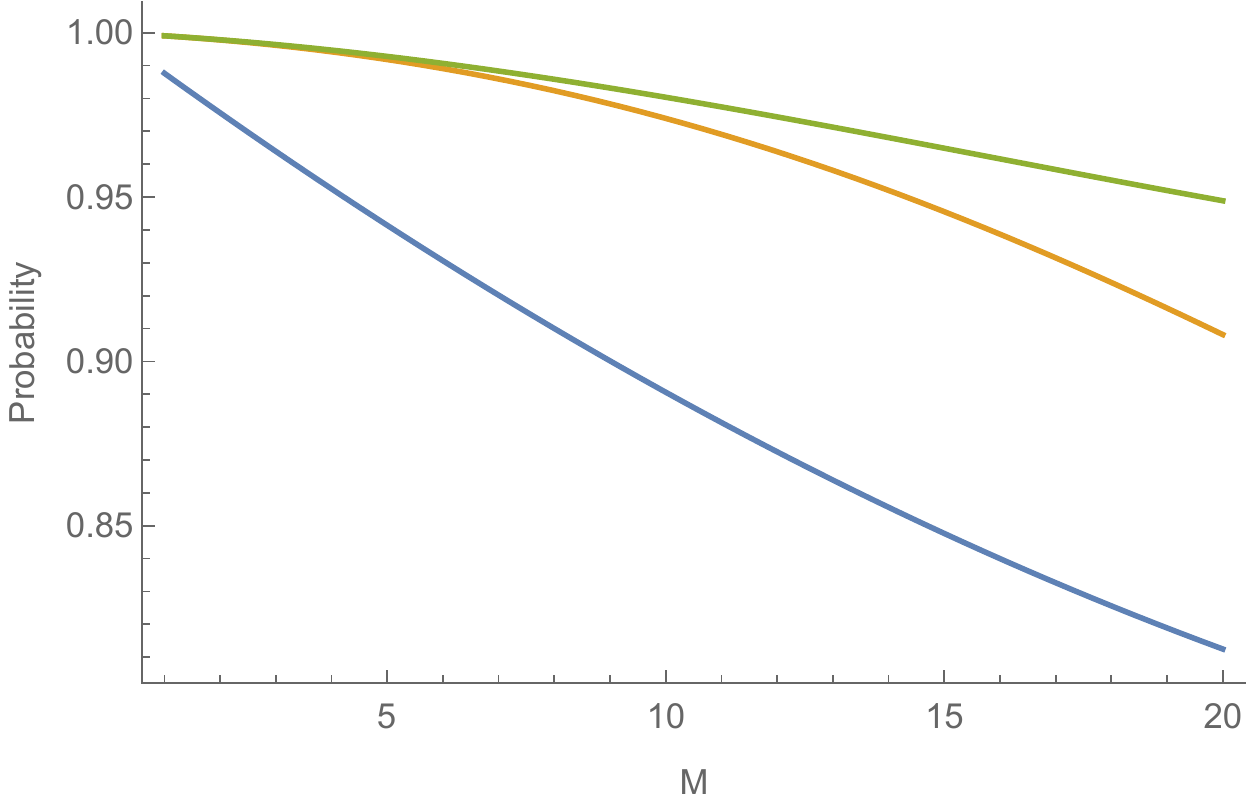}
		\caption{}
		\label{fig:Prob correct for phase systems c}
	\end{subfigure}
	
	\caption[Prob correct for phase systems]{Probability of obtaining the correct result as measured by the Mach-Zehnder interferometer set up as a function of the number of phase components $M$. Here a probability of $1$ corresponds to no error and the smaller the probability the larger the error. The blue line represents the no Error Averaging applied result, the orange line corresponds to the error when averaging across the entire system and the green line is the error when each component is averaged across individually. All three graphs were created with the variance of the error in a single phase shifter being $0.005\ \textrm{rad}^{2}$ and for (a) $N=2$, in (b) $N=4$ and in (c) $N=16$. \label{fig:Error-as-measured all}}
\end{figure}
Figure~\ref{fig:Error-as-measured all} shows how the error, as measured by looking at expected photon number values in the output port of a MZ interferometer, varies as the number of phase components increases as well as how the error changes with increasing Error Averaging, $N$. The behaviour as $N$ increases is as expected with the error close to disappearing for low $M$, that is, a small number of phase shifters in series, and $N=16$. Interestingly a difference between the two Error Averaging methods can be seen from $M\approx6$ onwards. This could either be suggesting an issue with the quality of the second order approximations, as seen in Figure~\ref{fig:Probability-of-success all} or that there is some more fundamental point at which there is a clear benefit to averaging each component individually. The next consideration was how the probability of success changes with $M$.

Figure~\ref{fig:Probability-of-success all} shows how the probability of success changes as the number of phase shifters in a series increases when averaging across the entire system as well as when averaging across each component individually. The effect of varying the amount of averaging is also shown for both the first and second order analytical solution. The top four graphs were plotted for a low value of the variance on the individual phase shifters. This was done so that the behaviour when the first and second order approximations diverge can be clearly seen. 

As the total number of components increases with both increasing $M$ and $N$, the probability of success decreases. However it does so at a decreasing rate which is important for scaling to large systems. The two methods of Error Averaging also show very similar behaviour in their overall trends although the variation between the first and second order approximations in the two encoding methods diverges. This is suggestive of a manifestation of the Zeno effect, whereby continuously correcting produces less variation than doing the same amount of correction at the end. First and second order solutions in the averaging over the entire system case diverge very early when compared with those for averaging every step. Interestingly it appears that the first order analytical approximation is suitable when averaging each component individually even for larger or equivalently higher error systems. This can most clearly be seen in Figure~\ref{fig:phase psuccess e} where the first order approximation diverges from the second order approximation almost instantly while in Figure~\ref{fig:phase psuccess f} the first and second order approximations both follow each other closely. It is again observed that as $N$ increases, the probability of success goes down. This could also be suggesting that the variation in the statistical simulation is also reduced, implying a greater amount of averaging reduces the variability in any given sample of the applied phase.
\begin{figure}
	\centering
	
	\begin{subfigure}[b]{.45\columnwidth}
		\includegraphics[width=\columnwidth]{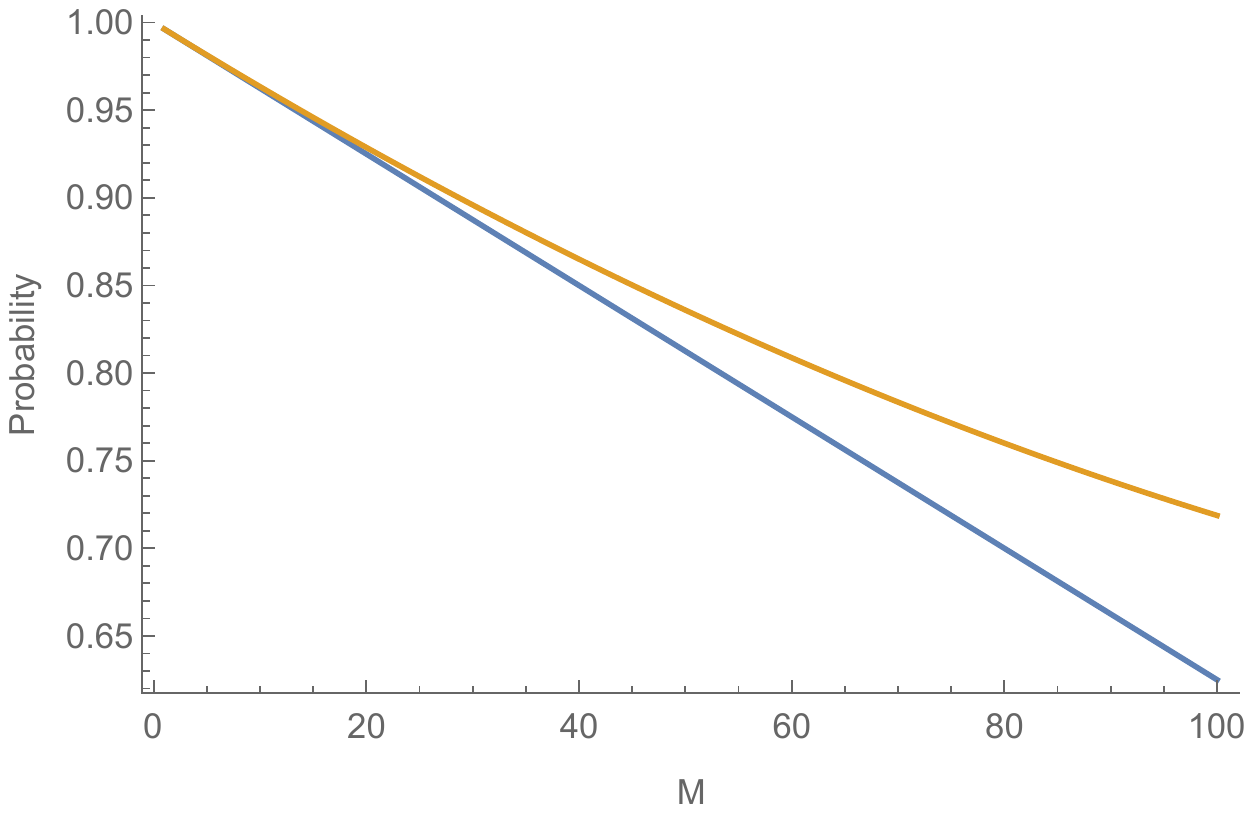}
		\caption{}\label{fig:phase psuccess a}
	\end{subfigure}
	\begin{subfigure}[b]{.45\columnwidth}
		\includegraphics[width=\columnwidth]{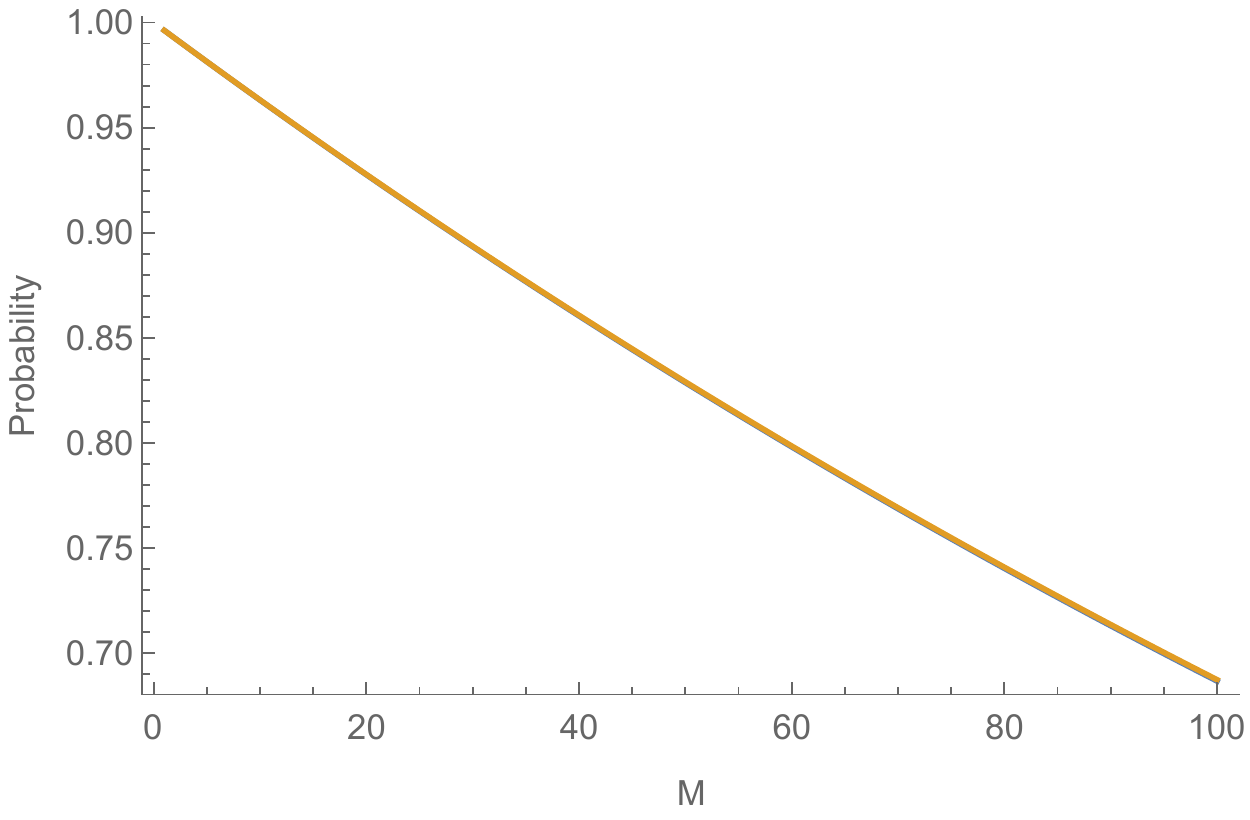}
		\caption{}\label{fig:phase psuccess b}
	\end{subfigure}

	\begin{subfigure}[b]{.45\columnwidth}
		\includegraphics[width=\columnwidth]{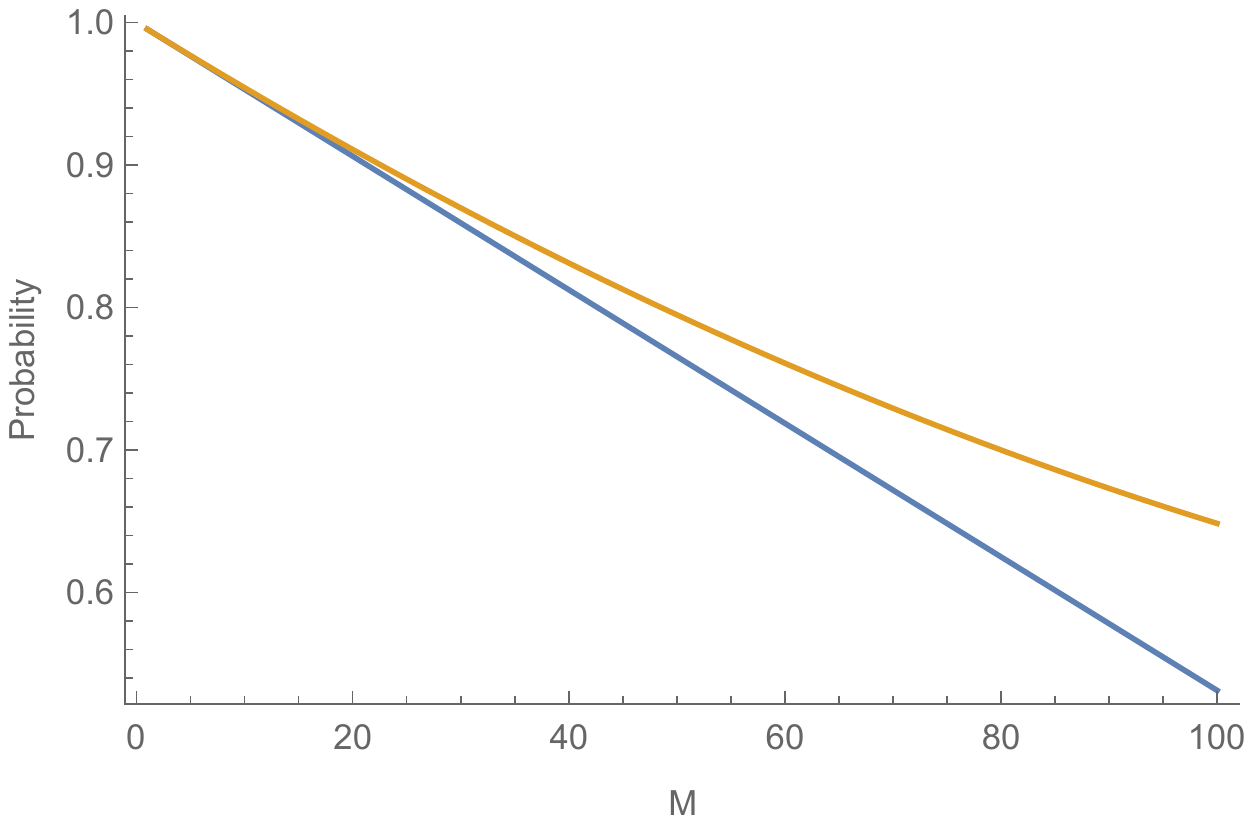}
		\caption{}\label{fig:phase psuccess c}
	\end{subfigure}
	\begin{subfigure}[b]{.45\columnwidth}
		\includegraphics[width=\columnwidth]{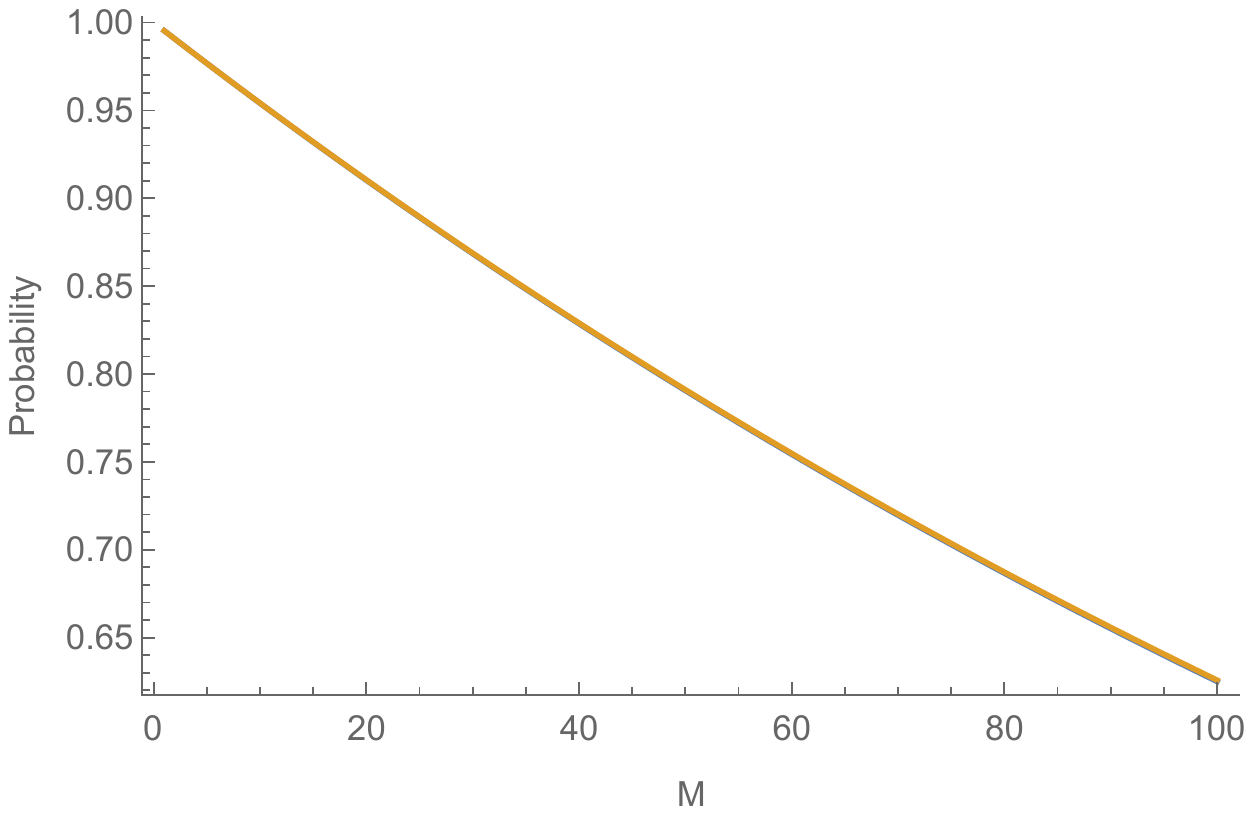}
		\caption{}\label{fig:phase psuccess d}
	\end{subfigure}

	\begin{subfigure}[b]{.45\columnwidth}
		\includegraphics[width=\columnwidth]{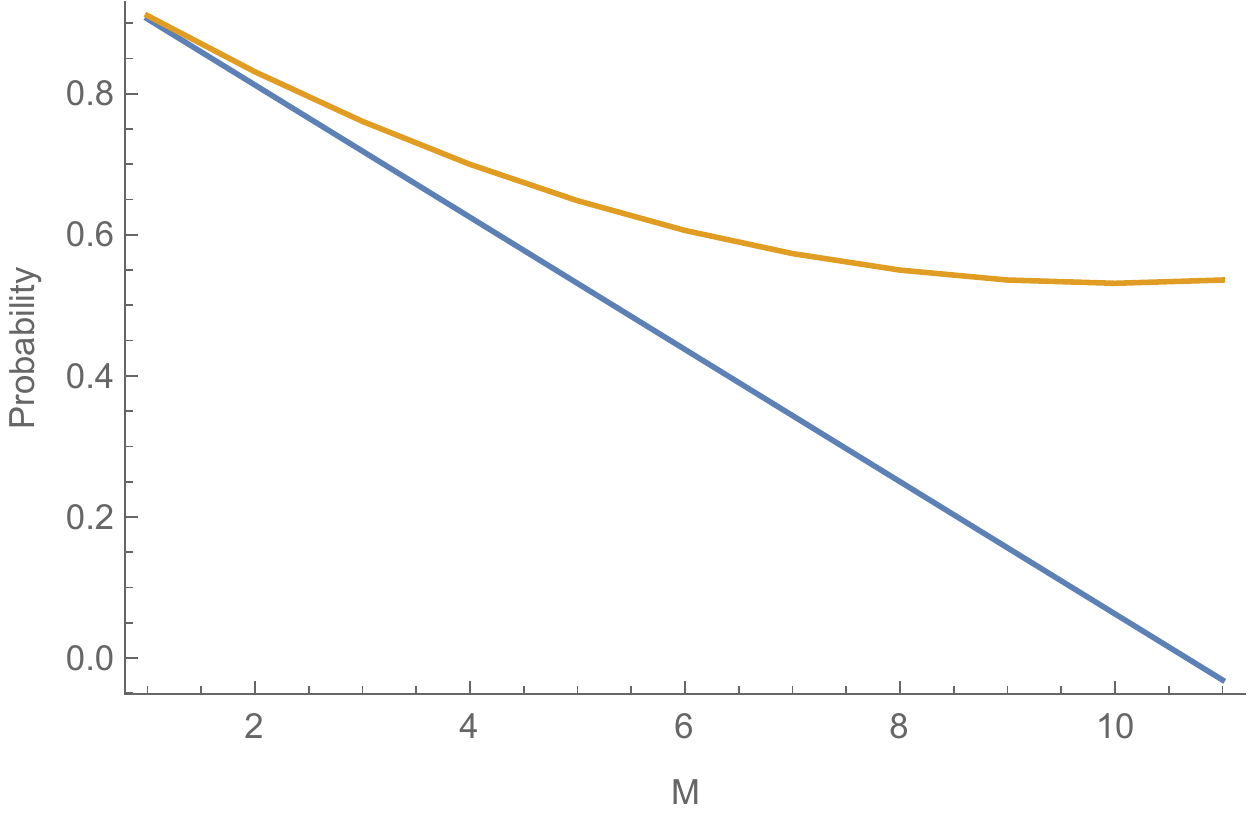}
		\caption{}\label{fig:phase psuccess e}
	\end{subfigure}
	\begin{subfigure}[b]{.45\columnwidth}
		\includegraphics[width=\columnwidth]{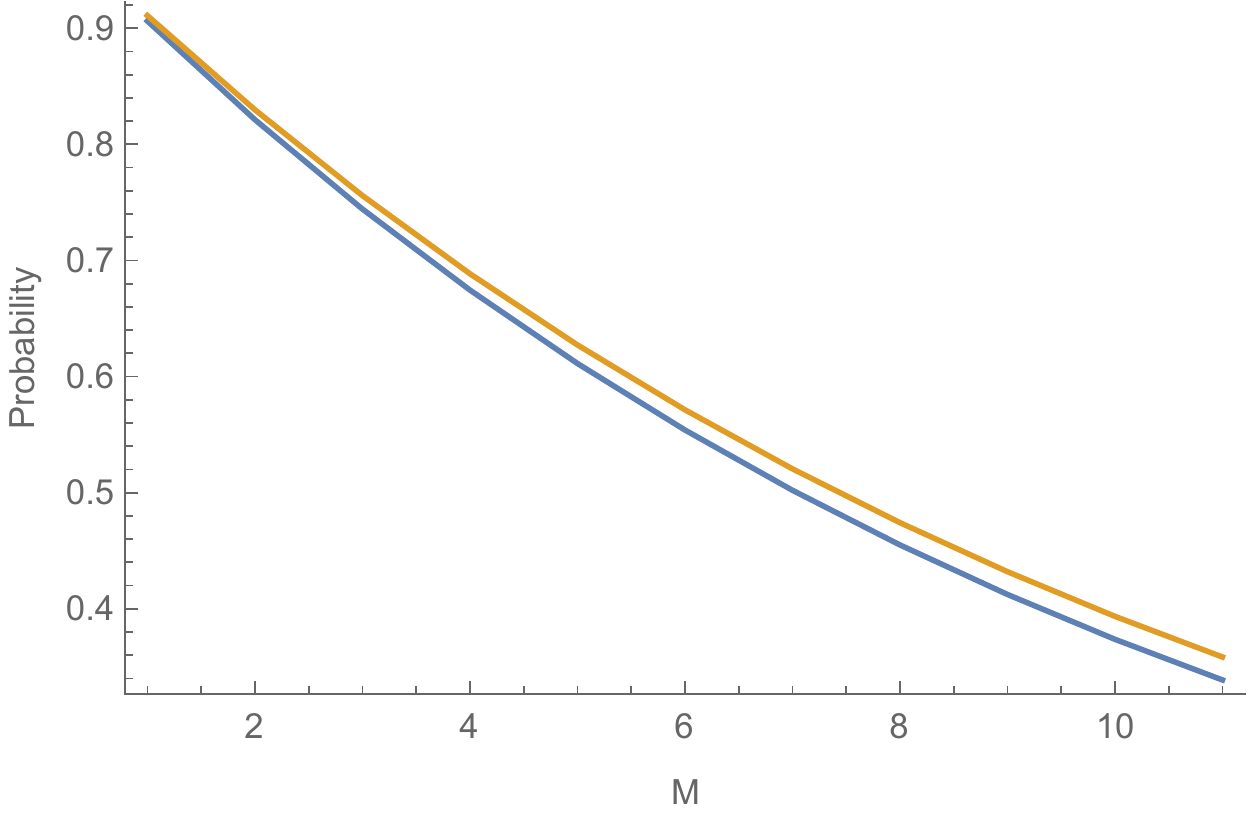}
		\caption{}\label{fig:phase psuccess f}
	\end{subfigure}
	
	\caption[Prob success for phase systems]{Probability of success as a function of the number of phase components $M$ for averaging across the system (left) and averaging across each component (right). The blue line is the first order approximation and orange is the second order approximation of the analytical value. a) and b) individual phase shifter variance is $0.005\ \textrm{rad}^{2}$ with $N=4$, c) and d) individual phase shifter variance is $0.005\ \textrm{rad}^{2}$ with $N=16$ and e) and f) individual phase shifter variance is $0.1\ \textrm{rad}^{2}$ with $N=16$. \label{fig:Probability-of-success all}}
\end{figure}


\subsection{Statistical Modelling of the Applied Phase\label{Statistical Modelling of the Applied Phase}}

Given the variability in the higher order terms for larger errors in can be concluded that, in general all orders need to be considered to fully understand the behaviour of the corrected systems. As this is intractable and to better understand the behaviour of the three phase applying systems, the total applied phase was modelled numerically using \textit{Mathematica} with phase values chosen from a Gaussian random distribution with mean $0$ and variance $v$. This corresponds to Equation~(\ref{eq:single parameter, single mode}) with $\theta=\prod_{i}\theta_{i}$. This was repeated $5000$ times and the results are shown in Figures \ref{fig:Total-applied-phase1} and \ref{fig:Total-applied-phase2}. This again shows a difference between averaging across the entire system and averaging at each step. The variability of the total applied phase is smaller when each phase shifter is corrected individually, an indication that averaging each step is more effective. By comparing Figure~\ref{fig:Total-applied-phase1} with Figure~\ref{fig:Error-as-measured all} at $M=15$ we can infer that the difference between the two error correction methods seen in Figure~\ref{fig:Error-as-measured all} is not entirely due to the quality of the approximations used in each case.
\begin{figure}
\centerline{\includegraphics[width=\columnwidth]{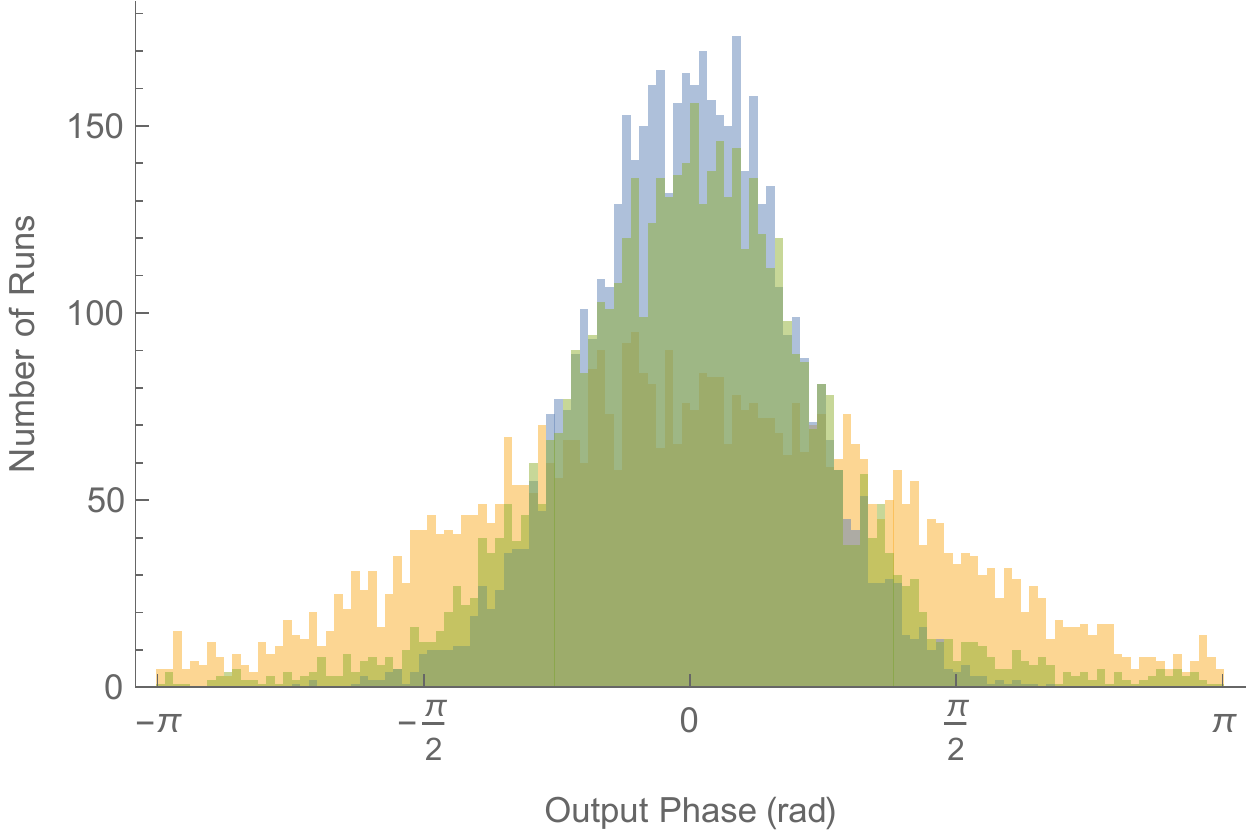}}
\caption{Top: Total applied phase over $5000$ runs for no averaging (orange), averaging across the entire system (green) and averaging each phase shifter individually (blue). Bottom: Histogram of the total applied phases. Each individual phase shifter has a variance of $0.1\ \textrm{rad}^{2}$ and each system has $15$ phase shifters in series. The two error averaged circuits are each averaged $4$ times. \label{fig:Total-applied-phase1}}
\end{figure}
\begin{figure}
\centerline{\includegraphics[width=\columnwidth]{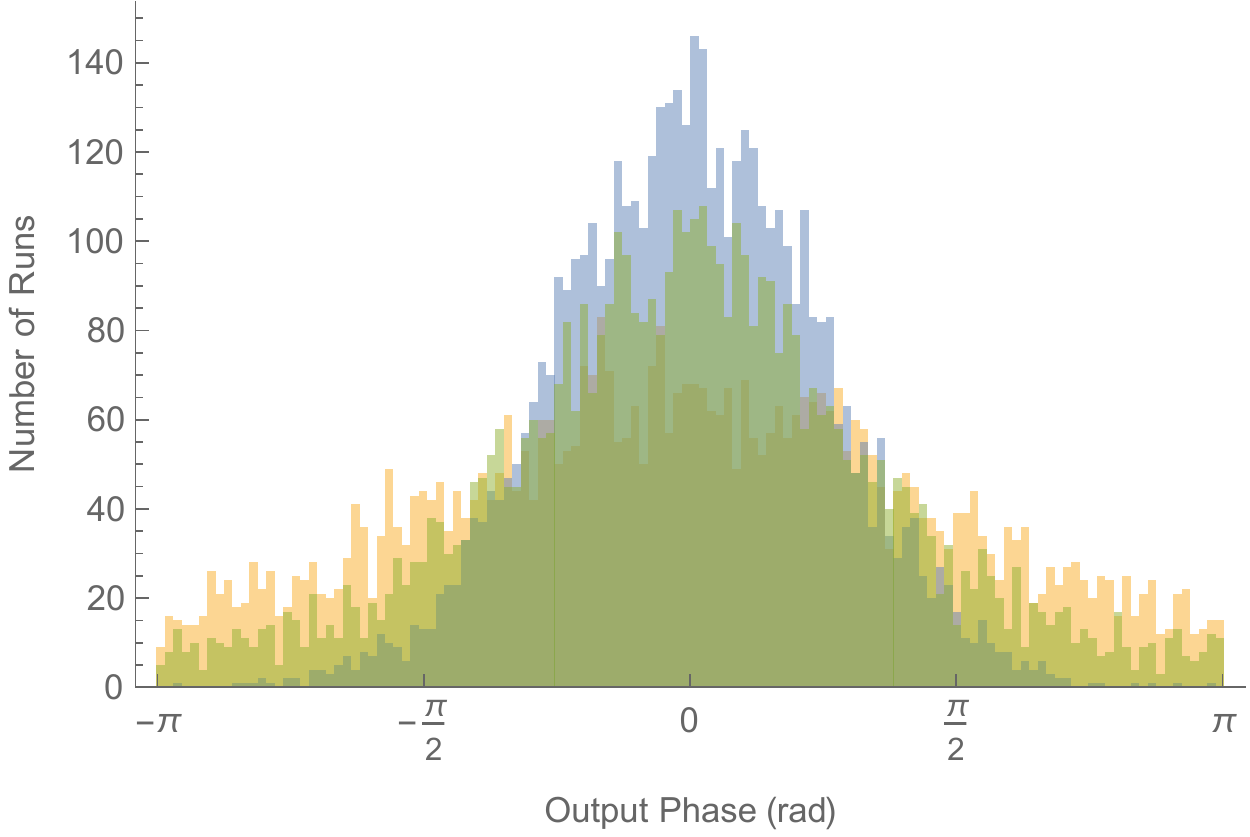}}
\caption{Top: Total applied phase over $5000$ runs for no averaging (orange), averaging across the entire system (green) and averaging each phase shifter individually (blue). Bottom: Histogram of the total applied phase. Each individual phase shifter has a variance of $0.3\ \textrm{rad}^{2}$ and each system has $8$ phase shifters in series. The two error averaged circuits are each averaged $4$ times. \label{fig:Total-applied-phase2}}
\end{figure}

The variance of the applied phase was estimated based on the statistical simulation of the total applied phase. Figure~\ref{fig:Variance-in-phase all} shows this variance as a function of $M$, the number of phase shifters in a series. Given the individual applied phases are uncorrelated they are expected to simply add, such that the variance without any Error Averaging is expected to be
\begin{equation}
\textrm{Total Variance}=vM\label{eq:Tot Var no correction}
\end{equation}
where $M$ is the number of phase shifters and $v$ is the variance in the individual phase shifter. The total variance when Error Averaging is similarly expected to be
\begin{equation}
\textrm{Total Variance}=\frac{vM}{N}\label{eq:Tot Var w/ correction}
\end{equation}
where $N$ is the number of times the system is averaged, again $N=1$ implies no averaging. As the phase is an angle with a finite range, this behaviour cannot hold for arbitrarily large $vM$. A completely random phase $\theta$ is still limited by the possible range of values, in our case chosen to be $-\pi<\theta\le\pi$. If the value of $\theta$ is indeed completely random then one will expect a uniform probability distribution of $P\left(\theta\right)=\frac{1}{2\pi}$. This then implies the maximum variance will be given by
\begin{eqnarray}
\textrm{Maximum variance} & = & \int_{-\pi}^{\pi}\theta^{2}P\left(\theta\right)\ d\theta\nonumber \\
& = & \frac{\pi^{2}}{3}\label{eq:Max Var}
\end{eqnarray}

\begin{figure}
\centerline{\includegraphics[width=\columnwidth]{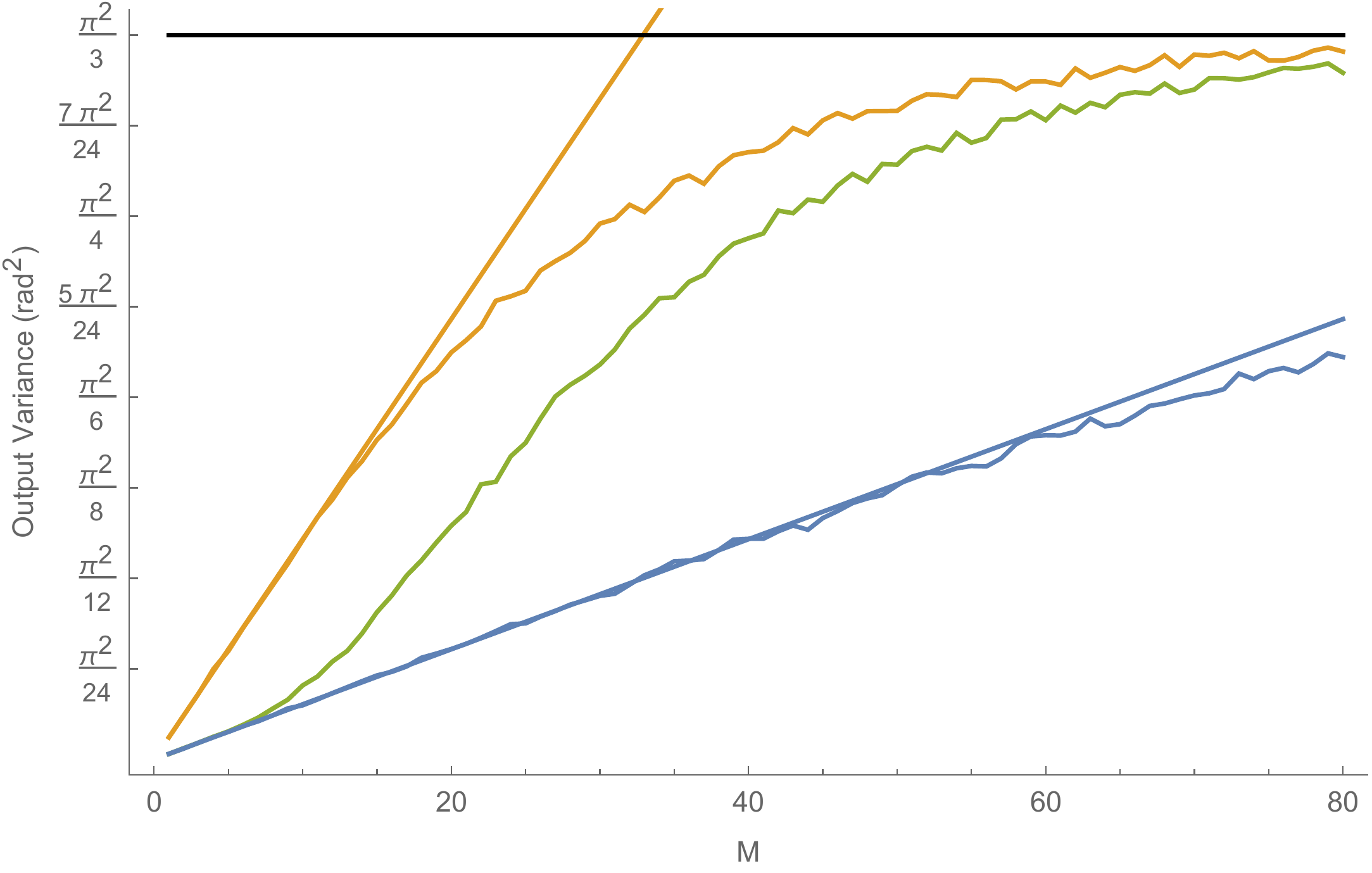}}
\caption{Variance in the total applied phase without Error Averaging (orange), when averaging across the entire system (green) and when averaging each component individually (blue), all plotted as a function of the number of phase shifters in series $(M)$. The variance of a single phase shifter is $0.1\ \textrm{rad}^{2}$ and the two error averaged systems averaged 4 times. The predicted, linear variance without any averaging (orange) and with averaging (blue) is also shown. These ignore the fact that the variance is actually the angular variance and so has some maximum allowable value given by Eq. \ref{eq:Max Var}, which is also shown in black. \label{fig:Variance-in-phase all}}
\end{figure}
Figure \ref{fig:Variance-in-phase all} shows that the two methods of error correction do indeed initially have the same effect. However the averaging across the system method departs from this linear regime from approximately $M=6$ after which it follows the general form of applying no correction. This suggests there is some limit to the total variation in a system Error Averaging can handle. This is not unexpected due to the limited domain for a phase shift or beam splitter ratio. The fact that averaging across the entire system mirrors the no averaging trend suggest that the positive effects of Error Averaging completely disappear in this regime. Averaging each step however does not appear to fall out of the linear regime. The lower slope at higher total output variance can be attributed to the variance approaching the maximum possible variance.  To determine if and when the averaging each step method of error correction fails, this process was repeated as a function of the variance in a single beam splitter. The total variance was determined from $50000$ data points for each value of $v$. Figure~\ref{fig:Variance(veriance)} demonstrated that again the two error correction methods initially are equivalent. Once more averaging across the system departs from the linear regime and now we can clearly see that so too does the averaging each step.

This is suggestive of the existence of some threshold for the amount of error in a system before Error Averaging fails to be beneficial. A single phase shifter with variance $v_{1}$ is effectively equivalent to $m$ phase shifters with individual variance $v_{2}=\frac{v_{1}}{m}$ if the entire system is being averaged across. This allows the phase error threshold to be estimated at about $0.5\ \textrm{rad}^{2}$ when averaging across each element and $\frac{0.5}{m}\textrm{rad}^{2}$ when averaging across a system of $m$ phase shifters. Explicitly, this suggests a phase variance threshold of $0.5\ \textrm{rad}^{2}$ within the corrected unitary. If each individual beam splitter has a variance of $0.1\ \textrm{rad}^{2}$, averaging across the system would be expected to be in the linear regime when $M\le5$ which is precisely what is seen in Figure \ref{fig:Variance-in-phase all}. These two thresholds obviously do not apply to a general error however it is hoped that with further study might reveal the values for such a threshold.
\begin{figure}
\centerline{\includegraphics[width=\columnwidth]{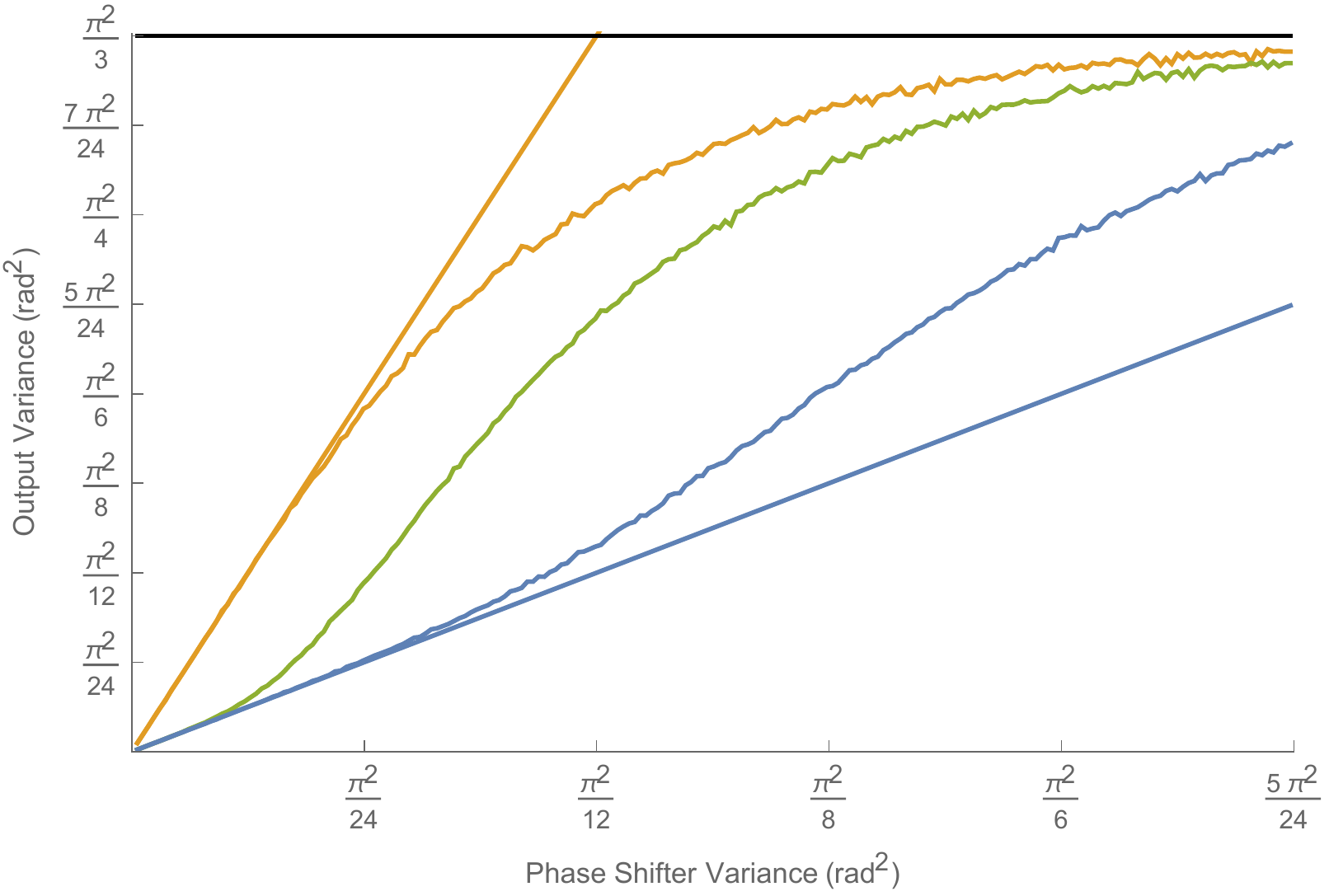}}
\caption{Variance in the total applied phase without Error Averaging (orange), when averaging across the entire system (green) and when averaging each component individually (blue), all plotted as a function of the variance in each individual phase shifter. Each system applied $4$ phase shifters in series, that is $M=4$, and the two error averaged systems averaged $4$ times, that is $N=4$. The predicted variance without any averaging (orange) and with averaging (blue) is also shown along with the maximum allowable variance. \label{fig:Variance(veriance)}}
\end{figure}

It can now be concluded that some hybrid method of Error Averaging would be most suitable in general. The entire system would need to be broken into $x$ smaller systems, which are independently averaged across. The specific value of $x$ would be such that the number of components in the system, $m$, is maximised while the total error within each subsystem is kept below the appropriate threshold.

\section{Four Mode Implementation Comparison \label{Four Mode Impementation Comparison}}

To gain a better understanding of how useful this method of Error Averaging actually is, a more complex system was also investigated along with both methods of redundant error correction. Specifically, a four mode system with four beam splitters, each implemented as above, that is each being its own MZ interferometer as shown in Figure \ref{fig:4 mode basis diagram}. Three different input states were chosen: a single photon input in the top mode and the vacuum state at all other modes$\left(\left|1,0,0,0\right\rangle \right)$, two photons, both in the top mode $\left(\left|2,0,0,0\right\rangle \right)$ and two photons spread across the top two modes $\left(\left|1,1,0,0\right\rangle \right)$. For simplicity the system was chosen to target the identity and error reduction was then applied using both implementations. That is, by averaging each beam splitter as done in section \ref{1 photon N arbitrary} and by concatenating the entire system in an interferometer as shown in Figure \ref{fig: averaging 4 mode diagram}. All results were found using \textit{Mathematica} to sample from the appropriate transformation matrix representing the system and then compute the second order approximation of the photon number expectation values and probabilities for $N=1, 2 \textrm{ and } 4$. All results can be found in Appendix \ref{Appendix full of results}

\begin{figure}[h]
\centerline{\includegraphics[width=\columnwidth]{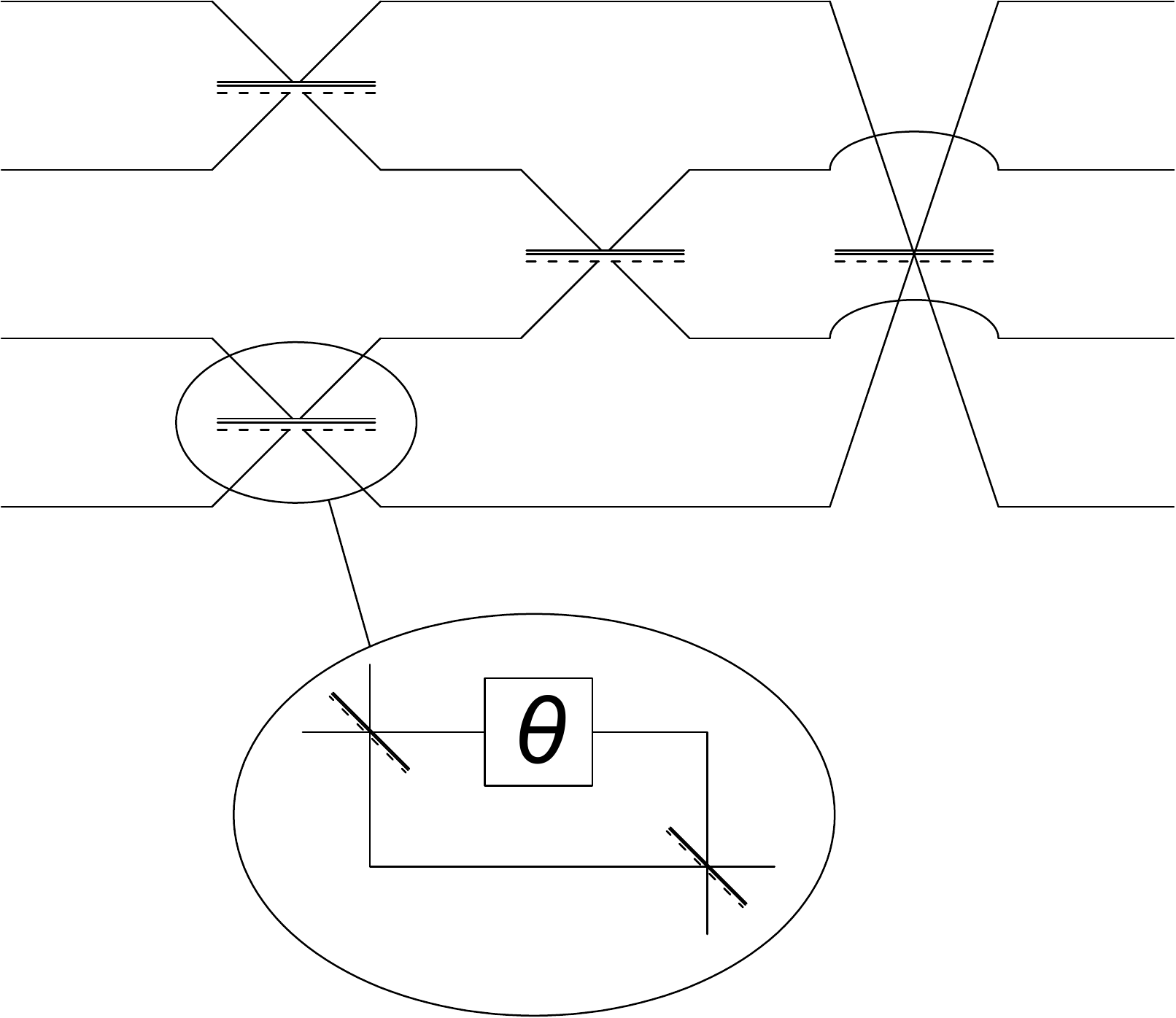}}
\caption{Diagram of the four mode linear optical network which forms the basis of the four mode set-ups. \label{fig:4 mode basis diagram}}
\end{figure}

\begin{figure}[h]
\centerline{\includegraphics[width=\columnwidth]{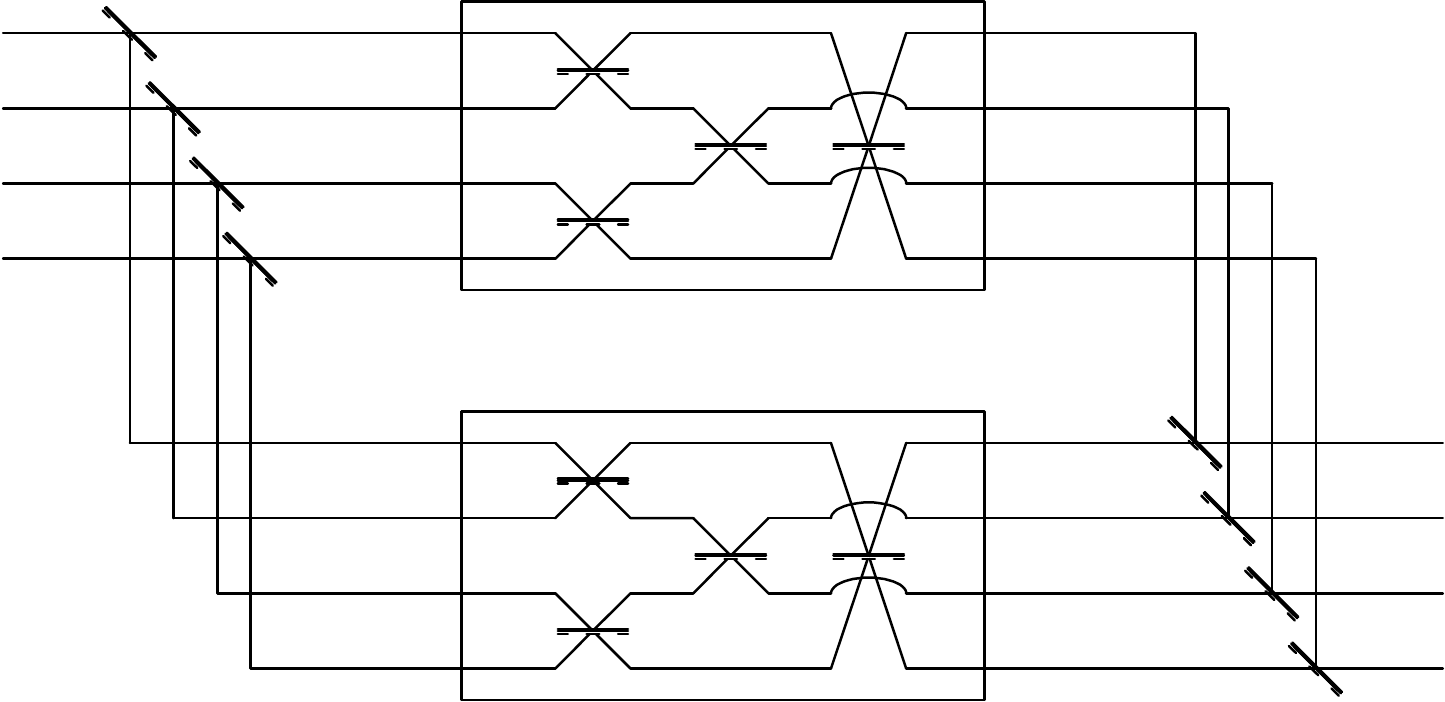}}
\caption{Diagram of the four mode linear optical network averaged across the system once. A two mode system averaged in this fashion could be considered as an Error Averaged dual rail single qubit unitary transformation. \label{fig: averaging 4 mode diagram}}
\end{figure}

These results show that the two correction methods produced equivalent results under the approximations used. The $1/N$ scaling in error after post selection, as seen above, was also observed, suggesting this pattern holds for higher numbers of modes and arbitrary system as expected from Theorem \ref{Theorem 1}.  The analysis that follows is based on the trends observed from these simulations.

\subsubsection{Generalising Example Results}

Starting with no error reduction, and given the correct output state $\left|\psi\right\rangle $ the following sequence can be defined. First defining the probability of obtaining the correct result when $N=1$ takes the form:
\begin{equation}
P_{1}(correct)=1-\frac{a_{1}}{b_{1}}v.
\end{equation}
The probability of an incorrect state is therefore
\begin{equation}
P_{1}(wrong)=\frac{a_{1}}{b_{1}}v.
\end{equation}
At this stage there are no error ports and so $P_{i}(wrong)+P_{i}(correct)$ represents the probability of success as previously defined. Note that the subscript indexes the corresponding number of averaging rounds. Explicitly $i+1$ corresponds to a system with twice as much averaging as $i$.  So, using the same observed form for the probability, averaging once changes these values to:
\begin{eqnarray}
P_{2}(correct) & = & 1-\frac{2a_{1}+1}{2b_{1}}v\nonumber \\
& \equiv & 1-\frac{a_{2}}{b_{2}}v
\end{eqnarray}
\begin{eqnarray}
P_{2}(wrong) & = & \frac{a_{1}}{2b_{1}}v\nonumber \\
& \equiv & \frac{a_{2}}{b_{2}}v
\end{eqnarray}
which can be further iterated.  In general
\begin{eqnarray}
P_{n}(correct) & = & 1-\frac{2a_{n-1}+1}{2b_{n-1}}v\nonumber \\
& = & 1-\frac{\left(2^{n-1}a_{1}+\left(2^{n-1}-1\right)\right)}{2^{n-1}b_{1}}v\\
P_{n}(wrong) & = & \frac{a_{n-1}}{2b_{n-1}}v\nonumber \\
& = & \frac{a_{1}}{2^{n-1}b_{1}}v
\end{eqnarray}
The probability of obtaining the correct state with post selection will then be
\begin{eqnarray}
&  & P_{n}\left(correct\left|\textrm{post selection}\right.\right)\nonumber \\
& = & \left(1-\frac{a_{1}}{2^{n-1}b_{1}}v\right)\xrightarrow[n\rightarrow\infty]{}1\label{eq:PcorrectGeneral}
\end{eqnarray}
The probability of success is 
\begin{equation}
P_{n}\left(success\right) =  1-\frac{\left(2^{n-1}-1\right)\left(a_{1}+1\right)}{2^{n-1}b_{1}}v\nonumber \\
\end{equation}
Hence we get an asympotitic expression for the probability of success to be 
\begin{equation}
\lim_{n\rightarrow\infty}P_{n}\left(success\right)=1-\left(\frac{a_{1}+1}{b_{1}}\right)v\label{eq:PsuccessGeneral}
\end{equation}
The result can be understood to be the first order approximation to Equation \ref{eq:approx commuting, general case}. The $\frac{a_1+1}{b_1}$ coefficient does not quite match what might be expected from Equation \ref{eq:approx commuting, general case} however one might only expect qualitative agreement as there is not any clear isomorphic map between the parameters in the system and the error coefficients of the Lie algebra generators.

This result hints at the self correcting nature of Error Averaging. By considering the inner corrected system with error laden beam splitters as the initial step in the sequence then each further step will be averaging across both the fixed beam splitters and the original error laden system. This could allow some of the beam splitters to be corrected  making the base assumptions on the quality of the fixed beam splitters less restrictive. This is highlighted in Figure \ref{fig:gen system} where, depending on which components are considered to be the system, it is averaged either $8$, $4$ or $2$ times. 

\section{Discussion\label{Discussion}}

The analysis presented in~\cref{implementation,averaging at end vs step,Four Mode Impementation Comparison} concentrated on implementing a single-mode phase shift either on it's own or as part of a Mach-Zender interferometer implementing a beam-splitter type transformation.  These could then be further used to build up higher-dimensional unitary transformations using any particular choice of decomposition, for which a specific decomposition was analysed in Section~\ref{Four Mode Impementation Comparison}.

On the other hand, one may want to redundantly encode an entire unitary rather than just the phases defining the internal parameters.  In this case we can use equation~\ref{eq:approx commuting, general case} and for simplicity consider the specific case with $T_l^2= I$
\begin{equation}
	M = U \prod_l e^{-\frac{1}{2}\sigma_l^2} = U e^{-\frac{1}{2} n^2 \sigma^2},
\end{equation}
where the $n^2$ term appears as the product over all $n^2$ generators.
This would result in an effective operator transformation for a $k$ photon state as
\begin{equation}
	\frac{{a^\dagger}_i^{k}}{\sqrt{k!}} \rightarrow 
	e^{-k n^2 \sigma^2/2} \frac{1}{\sqrt{k!}} \left(\sum_j U_{ij} {a^\dagger}_i \right)^k
\end{equation}
The coefficient here represents a reduction in the amplitude should this transformation be applied to a state, and hence represents the probability of success. To achieve a $O(1)$ probability of success, then the operator noise must obey $\sigma = O(k^{-1/2})$ and $\sigma = O(n^{-1})$ as $k,n \rightarrow \infty$.  These results are dependant on the assumptions and the desired performance, in terms of probability of success, will depend on the specific application.  However, it must be kept in mind that no error correction has been performed yet it is still possible to achieve a constant success probability with a reasonable scaling of the noise with respect to the network size. 

Within the constructions presented here some optical elements utilised have been assumed ideal.  In particular, the encoding beam-splitters were assumed to have exactly 50:50 reflectivities.  A more general consideration is that of the fault-tolerance of this encoding.  That is, can the ideal elements be error corrected whilst maintaining the error correcting power of the scheme.   In this paper we have focused on merely the error correcting power of different arrangements of phase shifts and how it varies across two choices of decomposition.  But there will be many and varied choices about how to implement fault tolerant constructions with some better than others, in much the same way as is applicable for discrete system in quantum computing implementations.  Never-the-less, the fact that, under an approximation of small errors the encoding tends to the ideal operation in the limit of large encoding sizes, it would be reasonable to expect that fault-tolerant constructions exist.  

\section{Comparison with conventional quantum error correction \label{Comparison with conventional error correction}}

It is insightful to qualitatively discuss the parallels between conventional qubit quantum error detection and correction techniques, and our error averaging technique. The simplest code to see this parallel is by considering the 3-qubit code, which is able to detect and correct at most a single physical bit-flip error on a 3-fold redundantly encoded logical qubit. In the 3-qubit code the logical qubit is encoded via GHZ-type entanglement across the three physical qubits using two maximally entangling CNOT gates. 
Specifically, encoding implements the redundant mapping $\alpha\ket{0} + \beta\ket{1} \to \alpha\ket{000} + \beta\ket{111}$ in the logical basis. In our scheme, on the other hand, the redundant encoding takes the form of W-like entanglement, implemented via an optical fanout operation, where a single excitation in a single mode is mapped to a superposition of a single excitation across multiple modes. Specifically, the encoding is of the form $\hat{a}_1^\dag \to \frac{1}{\sqrt{N}}(\hat{b}_1^\dag + \dots + \hat{b}_N^\dag)$. This is qualitatively very distinct from the previous GHZ-type encoding, since GHZ states are maximally-entangled states, whereas W-states are not. Unlike GHZ states, which collapse onto a perfectly mixed $N-1$ qubit state upon loss of just a single qubit, the loss of a single mode from a W-state preserves most entanglement for large $N$. This leads us to speculate that this property of W-states enables much of the structure of encoded states to be preserved upon localised errors. Indeed, for $N\gg 1$ we anticipate that the failure of a relatively small subset of the redundant operations will have little impact on the integrity of the entire encoded state, owing to this unique property of the structure of loss in W-states.

Like conventional quantum error correction, we observed error threshold behaviour in our analysis. That is, we are only able to improve the fidelity of a state if its initial fidelity is above an error correction threshold. Below this threshold the error correction technique fails to improve the state. Indeed, non-zero thresholds must necessarily apply so as not to violate the quantum no-cloning theorem.

We observed that a simple form of circuit concatenation, whereby the protocol is recursively embedded within itself to construct larger nested codes, enables higher degrees of error correction, asymptoting to some maximum. This is congruent with conventional codes, where code concatenation asymptotically improves error correction at the expense of increased physical resources to mediate the more complex encoding.

In our scheme the post-selection upon detecting no photons in the designated failure modes is equivalent to syndrome measurement in traditional qubit codes. Successful post-selection effectively projects the encoded state back into the codespace, whereas failure heralds an unsuccessful syndrome extraction, thereby mapping the unitary error to a located loss error. The probability of detecting no photons in the failure modes can be associated with the error detection probability in traditional codes, and the respective conditional probability of measuring the correct output state with the error correction probability.

An interesting open question is whether the structure of the redundant encoding we utilise in our protocol may be translated to other physical architectures or conventional qubit settings, or rather whether it is very specific to photonic linear optics.

\section{Conclusion\label{Conclusion}}

We have shown how, given multiple noisy copies of a linear optical unitary network, Error Averaging can be use to implement a transformation that tends towards the average with reduced variance at the cost of the probability of success.  After post-selection, Error Averaging forms a rudimentary error correction protocol by filtering the noise from the redundant copies of the unitary network. For this to form a true error correction protocol it will be necessary to introduce some sort of loss correction. The losses which will need to be corrected are unique however in that they are heralded and located, potentially simplifying the problem enormously. The variance in the transformations have been shown to scale as $\frac{1}{N}$ where $N$ represents the number of redundant copies of the network.  We have provided the mathematical basis necessary to determine the effect of Error Averaging on an arbitrary linear unitary and with fully characterised solutions for arbitrary single parameter noise and multiple parameter small Gaussian noise. We have also analytically determined the photon number expectation values for two mode systems with both one and two photon inputs, numerically simulated the output expectation values in four mode systems for both one and two photon inputs and numerically simulated the variance for different arrangements of phase shifters.

Two methods of Error Averaging for phase shifts have been presented which appear to have similar effects under certain conditions. In particular averaging after sequentially applying phases has the same behaviour as averaging each phase provided the errors are small. This behaviour is conjectured to be explained by considering the errors as approximately commuting.

\begin{acknowledgments}
	We thank Michael Bremner for motivating discussions. This research was funded by the Australian Research Council Centre of Excellence for Quantum Computation and Communication Technology (Project No.CE110001027). P.P.R is funded by an ARC Future Fellowship (project FT160100397).
\end{acknowledgments}

\bibliography{references}

\appendix
\begin{widetext}
\section{Four Mode Numerical Results \label{Appendix full of results}}
\FloatBarrier
This appendix contains all simulation results for the four mode system discussed in Section \ref{Four Mode Impementation Comparison}. All results are based on a second order Taylor expansion with $\nu\ll1$.

Tables \ref{tab:1 photon output prob bs} and \ref{tab:1 photon output prob as}
show the output probabilities and correct result probability with
post selection for the single photon input state $\left|1,0,0,0\right\rangle $.
These show that, at least for a single photon the two correction methods
are equivalent. We also see the halving of errors as seen in sections
3.2 and 3.4 suggesting this pattern may hold for a single photon with
an arbitrary number of modes.

\begin{table}
\resizebox{\textwidth}{!}{

\begin{centering}
	\begin{tabular}{|c|c|>{\centering}p{4cm}|>{\centering}p{4cm}|}
		\hline 
		Output State & No Error Reduction & Averaging Beam splitters Once $\left(N=2\right)$ & Averaging Beam splitters Twice $\left(N=4\right)$\tabularnewline
		\hline 
		\hline 
		$\left|1,0,0,0\right\rangle $ & $1-\frac{v}{2}$ & $1-\frac{3v}{4}$ & $1-\frac{7v}{8}$\tabularnewline
		\hline 
		$\left|0,1,0,0\right\rangle $ & $\frac{v}{4}$ & $\frac{v}{8}$ & $\frac{v}{16}$\tabularnewline
		\hline 
		$\left|0,0,1,0\right\rangle $ & $0$ & $0$ & $0$\tabularnewline
		\hline 
		$\left|0,0,0,1\right\rangle $ & $\frac{v}{4}$ & $\frac{v}{8}$ & $\frac{v}{16}$\tabularnewline
		\hline 
		$\left|1,0,0,0\right\rangle $ with post selection & $1-\frac{v}{2}$ & $1-\frac{v}{4}$ & $1-\frac{v}{8}$\tabularnewline
		\hline 
	\end{tabular}
	\par\end{centering}

}

\caption[Output probabilities for various levels of correcting the individual
beam splitters in a 4 mode set-up given an input of the state $\left|1,0,0,0\right\rangle $
.]{Output probabilities for various levels of correcting the individual
beam splitters in a 4 mode set-up given an input of the state $\left|1,0,0,0\right\rangle $
where $v$ is the variance of the phase error . \label{tab:1 photon output prob bs}}
\end{table}
\begin{table}
\resizebox{\textwidth}{!}{

\begin{centering}
	\begin{tabular}{|c|c|>{\centering}p{4cm}|>{\centering}p{4cm}|}
		\hline 
		Output State & No Error Reduction & Averaging Across the System Once $\left(N=2\right)$ & Averaging Across the System Twice $\left(N=4\right)$\tabularnewline
		\hline 
		\hline 
		$\left|1,0,0,0\right\rangle $ & $1-\frac{v}{2}$ & $1-\frac{3v}{4}$ & $1-\frac{7v}{8}$\tabularnewline
		\hline 
		$\left|0,1,0,0\right\rangle $ & $\frac{v}{4}$ & $\frac{v}{8}$ & $\frac{v}{16}$\tabularnewline
		\hline 
		$\left|0,0,1,0\right\rangle $ & $0$ & $0$ & $0$\tabularnewline
		\hline 
		$\left|0,0,0,1\right\rangle $ & $\frac{v}{4}$ & $\frac{v}{8}$ & $\frac{v}{16}$\tabularnewline
		\hline 
		$\left|1,0,0,0\right\rangle $ with post selection & $1-\frac{v}{2}$ & $1-\frac{v}{4}$ & $1-\frac{v}{8}$\tabularnewline
		\hline 
	\end{tabular}
	\par\end{centering}

}

\caption[Output probabilities for various levels of correcting the across the
system in a 4 mode set-up given an input of the state $\left|1,0,0,0\right\rangle $.]{Output probabilities for various levels of correcting the across
the system in a 4 mode set-up given an input of the state $\left|1,0,0,0\right\rangle $
where $v$ is the variance of the phase error. \label{tab:1 photon output prob as}}
\end{table}

Tables \ref{tab:2 photon output prob bs} and \ref{tab:2 photon output prob as}
show the output probabilities and correct result probability with
post selection for the single photon input state $\left|2,0,0,0\right\rangle $.
What we see is, unsurprisingly, much the same as in the single photon
case with a heightened susceptibility to the error. This includes
the halving pattern however this is expected as adding two photons
in the same mode will not necessarily lead to new interference effects
being observed.

\begin{table}
\resizebox{\textwidth}{!}{

\begin{centering}
	\begin{tabular}{|c|>{\centering}p{4cm}|>{\centering}p{4cm}|}
		\hline 
		Output State & No Error Reduction & Averaging Beam splitters Once $\left(N=2\right)$\tabularnewline
		\hline 
		\hline 
		$\left|2,0,0,0\right\rangle $ & $1-v$ & $1-\frac{3v}{2}$\tabularnewline
		\hline 
		$\left|0,2,0,0\right\rangle $ & $0$ & $0$\tabularnewline
		\hline 
		$\left|0,0,2,0\right\rangle $ & $0$ & $0$\tabularnewline
		\hline 
		$\left|0,0,0,2\right\rangle $ & $0$ & $0$\tabularnewline
		\hline 
		$\left|1,1,0,0\right\rangle $ & $\frac{v}{2}$ & $\frac{v}{4}$\tabularnewline
		\hline 
		$\left|1,0,1,0\right\rangle $ & $0$ & $0$\tabularnewline
		\hline 
		$\left|1,0,0,1\right\rangle $ & $\frac{v}{2}$ & $\frac{v}{4}$\tabularnewline
		\hline 
		$\left|0,1,1,0\right\rangle $ & $0$ & $0$\tabularnewline
		\hline 
		$\left|0,1,0,1\right\rangle $ & $0$ & $0$\tabularnewline
		\hline 
		$\left|0,0,1,1\right\rangle $ & $0$ & $0$\tabularnewline
		\hline 
		$\left|2,0,0,0\right\rangle $ with post selection & $1-v$ & $1-\frac{v}{2}$\tabularnewline
		\hline 
	\end{tabular}
	\par\end{centering}

}

\caption[Output probabilities for various levels of correcting the individual
beam splitters in a 4 mode set-up given an input of the state $\left|2,0,0,0\right\rangle $.]{Output probabilities for various levels of correcting the individual
beam splitters in a 4 mode set-up given an input of the state $\left|2,0,0,0\right\rangle $
where $v$ is the variance of the phase error. \label{tab:2 photon output prob bs}}
\end{table}
\begin{table}
\resizebox{\textwidth}{!}{

\begin{centering}
	\begin{tabular}{|c|>{\centering}p{4cm}|>{\centering}p{4cm}|>{\centering}p{4cm}|}
		\hline 
		Output State & No Error Reduction & Averaging Beam splitters Once $\left(N=2\right)$ & Averaging Beam splitters Twice $\left(N=4\right)$\tabularnewline
		\hline 
		\hline 
		$\left|2,0,0,0\right\rangle $ & $1-v$ & $1-\frac{3v}{2}$ & $1-\frac{7v}{4}$\tabularnewline
		\hline 
		$\left|0,2,0,0\right\rangle $ & $0$ & $0$ & $0$\tabularnewline
		\hline 
		$\left|0,0,2,0\right\rangle $ & $0$ & $0$ & $0$\tabularnewline
		\hline 
		$\left|0,0,0,2\right\rangle $ & $0$ & $0$ & $0$\tabularnewline
		\hline 
		$\left|1,1,0,0\right\rangle $ & $\frac{v}{2}$ & $\frac{v}{4}$ & $\frac{v}{8}$\tabularnewline
		\hline 
		$\left|1,0,1,0\right\rangle $ & $0$ & $0$ & $0$\tabularnewline
		\hline 
		$\left|1,0,0,1\right\rangle $ & $\frac{v}{2}$ & $\frac{v}{4}$ & $\frac{v}{8}$\tabularnewline
		\hline 
		$\left|0,1,1,0\right\rangle $ & $0$ & $0$ & $0$\tabularnewline
		\hline 
		$\left|0,1,0,1\right\rangle $ & $0$ & $0$ & $0$\tabularnewline
		\hline 
		$\left|0,0,1,1\right\rangle $ & $0$ & $0$ & $0$\tabularnewline
		\hline 
		$\left|2,0,0,0\right\rangle $ with post selection & $1-v$ & $1-\frac{v}{2}$ & $1-\frac{v}{4}$\tabularnewline
		\hline 
	\end{tabular}
	\par\end{centering}

}

\caption[Output probabilities for various levels of correcting the across the
system in a 4 mode set-up given an input of the state $\left|2,0,0,0\right\rangle $.]{Output probabilities for various levels of correcting the across
the system in a 4 mode set-up given an input of the state $\left|2,0,0,0\right\rangle $
where $v$ is the variance of the phase error. \label{tab:2 photon output prob as}}
\end{table}

Tables \ref{tab:1,1 photon output prob as} and \ref{tab:1,1 photon output prob as}
show the output probabilities and correct result probability with
post selection for the single photon input state $\left|1,1,0,0\right\rangle $.
It can once more be seen that the two methods of error correction
appear to be equivalent. Now there is an underlying pattern clearly
forming which appears to hold for arbitrary one and two photon inputs.
This is important as it allows us to conclude about when it is most
useful to use each type of correction. It also allowed a prediction
of the error models for applications of Error Averaging, as discussed
below.

\begin{table}
\resizebox{\textwidth}{!}{

\begin{centering}
	\begin{tabular}{|c|>{\centering}p{4cm}|>{\centering}p{4cm}|>{\centering}p{4cm}|}
		\hline 
		Output State & No Error Reduction & Averaging Beam splitters Once $\left(N=2\right)$ & Averaging Beam splitters Twice $\left(N=4\right)$\tabularnewline
		\hline 
		\hline 
		$\left|2,0,0,0\right\rangle $ & $\frac{v}{2}$ & $\frac{v}{4}$ & $\frac{v}{8}$\tabularnewline
		\hline 
		$\left|0,2,0,0\right\rangle $ & $\frac{v}{2}$ & $\frac{v}{4}$ & $\frac{v}{8}$\tabularnewline
		\hline 
		$\left|0,0,2,0\right\rangle $ & $0$ & $0$ & $0$\tabularnewline
		\hline 
		$\left|0,0,0,2\right\rangle $ & $0$ & $0$ & $0$\tabularnewline
		\hline 
		$\left|1,1,0,0\right\rangle $ & $1-\frac{3v}{2}$ & $1-\frac{7v}{4}$ & $1-\frac{15v}{8}$\tabularnewline
		\hline 
		$\left|1,0,1,0\right\rangle $ & $\frac{v}{2}$ & $\frac{v}{8}$ & $\frac{v}{16}$\tabularnewline
		\hline 
		$\left|1,0,0,1\right\rangle $ & $0$ & $0$ & $0$\tabularnewline
		\hline 
		$\left|0,1,1,0\right\rangle $ & $0$ & $0$ & $0$\tabularnewline
		\hline 
		$\left|0,1,0,1\right\rangle $ & $\frac{v}{2}$ & $\frac{v}{8}$ & $\frac{v}{16}$\tabularnewline
		\hline 
		$\left|0,0,1,1\right\rangle $ & $0$ & $0$ & $0$\tabularnewline
		\hline 
		$\left|1,1,0,0\right\rangle $ with post selection & $1-\frac{3v}{2}$ & $1-\frac{3v}{4}$ & $1-\frac{3v}{8}$\tabularnewline
		\hline 
	\end{tabular}
	\par\end{centering}

}

\caption[Output probabilities for various levels of correcting the individual
beam splitters in a 4 mode set-up given an input of the state $\left|1,1,0,0\right\rangle $.]{Output probabilities for various levels of correcting the individual
beam splitters in a 4 mode set-up given an input of the state $\left|1,1,0,0\right\rangle $
where $v$ is the variance of the phase error. \label{tab:1,1 photon output prob bs}}
\end{table}
\begin{table}
\resizebox{\textwidth}{!}{

\begin{centering}
	\begin{tabular}{|c|>{\centering}p{4cm}|>{\centering}p{4cm}|>{\centering}p{4cm}|}
		\hline 
		Output State & No Error Reduction & Averaging Beam splitters Once $\left(N=2\right)$ & Averaging Beam splitters Twice $\left(N=4\right)$\tabularnewline
		\hline 
		\hline 
		$\left|2,0,0,0\right\rangle $ & $\frac{v}{2}$ & $\frac{v}{4}$ & $\frac{v}{8}$\tabularnewline
		\hline 
		$\left|0,2,0,0\right\rangle $ & $\frac{v}{2}$ & $\frac{v}{4}$ & $\frac{v}{8}$\tabularnewline
		\hline 
		$\left|0,0,2,0\right\rangle $ & $0$ & $0$ & $0$\tabularnewline
		\hline 
		$\left|0,0,0,2\right\rangle $ & $0$ & $0$ & $0$\tabularnewline
		\hline 
		$\left|1,1,0,0\right\rangle $ & $1-\frac{3v}{2}$ & $1-\frac{7v}{4}$ & $1-\frac{15v}{8}$\tabularnewline
		\hline 
		$\left|1,0,1,0\right\rangle $ & $\frac{v}{2}$ & $\frac{v}{8}$ & $\frac{v}{16}$\tabularnewline
		\hline 
		$\left|1,0,0,1\right\rangle $ & $0$ & $0$ & $0$\tabularnewline
		\hline 
		$\left|0,1,1,0\right\rangle $ & $0$ & $0$ & $0$\tabularnewline
		\hline 
		$\left|0,1,0,1\right\rangle $ & $\frac{v}{2}$ & $\frac{v}{8}$ & $\frac{v}{16}$\tabularnewline
		\hline 
		$\left|0,0,1,1\right\rangle $ & $0$ & $0$ & $0$\tabularnewline
		\hline 
		$\left|1,1,0,0\right\rangle $ with post selection & $1-\frac{3v}{2}$ & $1-\frac{3v}{4}$ & $1-\frac{3v}{8}$\tabularnewline
		\hline 
	\end{tabular}
	\par\end{centering}

}

\caption[Output probabilities for various levels of correcting the across the
system in a 4 mode set-up given an input of the state $\left|1,1,0,0\right\rangle $.]{Output probabilities for various levels of correcting the across
the system in a 4 mode set-up given an input of the state $\left|1,1,0,0\right\rangle $
where $v$ is the variance of the phase error. \label{tab:1,1 photon output prob as}}
\end{table}
\end{widetext}
\end{document}